\pdfoutput=1
\documentclass[11pt,a4paper]{article}

\usepackage[a4paper,margin=1in]{geometry}

\usepackage[T1]{fontenc}
\usepackage[utf8]{inputenc}
\usepackage{lmodern}

\usepackage{amsmath,amsfonts,amssymb}
\usepackage{graphicx}
\usepackage{xcolor}

\usepackage[numbers,sort&compress]{natbib}
\usepackage{amsthm}
\usepackage{pgf,tikz}
\usetikzlibrary{arrows.meta, shadows, positioning}
\usepackage{subcaption}
\usepackage{enumitem}

\usepackage{hyperref}
\hypersetup{
	citecolor=blue,
	linkcolor=blue,
	urlcolor=blue,
	pdftitle={Invariant-measure regularity for one-dimensional PDMPs with
		applications to shape transitions in run-and-tumble particle systems},
	pdfauthor={Leo Hahn},
	pdfsubject={math.PR},
	pdfkeywords={piecewise-deterministic Markov process, invariant measure,
		regularity, run-and-tumble particle, shape transition}
}

\theoremstyle{plain}
\newtheorem{Thm}{Theorem}
\newtheorem{Prop}{Proposition}
\newtheorem{Cor}[Prop]{Corollary}
\newtheorem{Lem}[Prop]{Lemma}

\theoremstyle{definition}
\newtheorem{Def}[Prop]{Definition}
\newtheorem{Cex}[Prop]{Counterexample}
\newtheorem{Ass}{Assumption}

\theoremstyle{remark}
\newtheorem{Rem}[Prop]{Remark}

\newcommand{\R}{^\mathrm{r}}
\newcommand{\sgn}{\text{\normalfont sgn }}

\begin{document}

\title{Invariant-measure regularity for one-dimensional PDMPs with
applications to shape transitions in run-and-tumble particle systems}
\author{Leo Hahn\thanks{%
	Institut de Math\'ematiques, Universit\'e de Neuch\^atel,
	Rue Emile-Argand 11, 2000 Neuch\^atel, Switzerland.\newline
	E-mail: \href{mailto:leo.hahn@unine.ch}{leo.hahn@unine.ch}.\newline
	ORCID: \href{https://orcid.org/0000-0002-9968-5913}{0000-0002-9968-5913}.}}
\date{\today}

\maketitle

\begin{abstract}
We study the regularity of the invariant density of 1D piecewise-deterministic Markov processes. Such densities are regular away from the zeros of the driving vector fields, but may become singular at these points, a phenomenon so far understood only when a single field vanishes. In the case of several simultaneously vanishing vector fields, we identify a new threshold on the jump rates separating continuity from divergence, which emerges from the interplay of all vanishing fields. Away from the zeros of the vector fields, we turn to distribution theory and the Kolmogorov forward equation. We show $C^r$ regularity of the density on intervals where the vector fields are $C^r$ and do not vanish, which is optimal, allowing $C^{r-1}$ position-dependent jump rates, as well as resetting contributing a $C^{r-1}$ source term to the Kolmogorov forward equation. For $r = 0$, it suffices that the jump rates be measurable and the resetting source be absolutely continuous with respect to Lebesgue measure. As an application, we study the shape transition of two run-and-tumble particle systems whose invariant measures cannot be computed explicitly, showing that shape transition does not depend on exact solvability.
\end{abstract}

\begin{center}
\begin{minipage}{0.85\textwidth}

\end{minipage}
\end{center}

\section{Introduction}

Run-and-tumble particles alternate between periods of uniform linear motion (runs) and rapid, random reorientation (tumbles)~\cite{berg04}. Describing the motion of bacteria~\cite{berg72} and algae~\cite{bennett13}, they now constitute a standard minimal model in active matter~\cite{ramaswamy10,marchetti13,bechinger16}. In particular, in one dimension, they provide a valuable setting in which non-equilibrium phenomena, such as accumulation at boundaries~\cite{elgeti15} and motility-induced phase separation~\cite{cates15}, can be analyzed by explicitly computing the invariant measure~\cite{slowman16,malakar18,angelani19,das20}.

This invariant measure already differs sharply from the Boltzmann distribution of the corresponding reversible dynamics for a single RTP with two velocities in a confining potential. It is supported on a compact interval $[x_-,x_+]$ and its density is continuous when the tumble rate is high, but diverges at $x_\pm$ when the tumble rate is low~\cite{dhar19}. This qualitative change, known as \emph{shape transition}, separates a close-to-equilibrium regime from a far-from-equilibrium one, in the spirit of the universality classes of~\cite{hahn23} and of the qualitative changes displayed by the invariant measure in~\cite{hahn24}. Singularities may also appear in the interior of the support, as they do in three-velocity models for the first coordinate of a higher-dimensional RTP~\cite{basu20}, non-instantaneous tumbles~\cite{sun25} or the separation of two interacting RTPs~\cite{ledoussal21}. It is important to note, however, that all existing quantitative shape transition results rely on the explicit computation of the invariant measure, which seems intractable for most systems with more than two velocities. \\

The study of shape transition is part of the larger question of regularity for the invariant measure of dynamical systems with random switching $X_t = (x_t, \sigma_t)$~\cite{bakhtin12,bakhtin15}. In this subclass of piecewise-deterministic Markov processes (PDMPs)~\cite{davis93,malrieu16}, the position $x_t \in \mathbb R^d$ follows the ODE $\partial_t x_t = v_{\sigma_t}(x_t)$ while the index $\sigma_t \in \Sigma$ jumps between finitely many values. The jump rate at which $\sigma_t$ leaves the state $\sigma$ is denoted $\lambda_\sigma$. Under Hörmander bracket conditions at an accessible point, the invariant measure $\pi$ of these processes admits a density~\cite{bakhtin12,benaim15,BTK16}, i.e.~$\pi = \sum_{\sigma \in \Sigma} p_\sigma(x)dx \otimes \delta_\sigma$. In line with the concept of shape transition, detailed analysis of specific models reveals that the invariant density can exhibit different behaviors, being smooth in some cases~\cite{bakhtin18} and developing singularities in others~\cite{bakhtin21}. High jump rates have recently been shown to ensure global regularity~\cite{benaim23,benaim24} and specific settings leading to singularities at low jump rates have also been identified~\cite{benaim22}. However, a detailed, local understanding of regularity is still missing for arbitrary fixed jump rates. In the one-dimensional setting, the picture is much clearer~\cite{balazs11,bakhtin15}. 
The invariant densities can only form singularities at the zeros of the $v_\sigma$, where their appearance is decided by a competition between the deterministic contraction and the stochastic jumps: at points $x_0$ where a single field $v_{\sigma_0}$ vanishes, $p_{\sigma_0}$ is continuous when $\lambda_{\sigma_0} > -v_{\sigma_0}'(x_0)$ and diverges when $\lambda_{\sigma_0} < -v_{\sigma_0}'(x_0)$. \\

For the relative position of generic interacting RTPs, multiple $v_\sigma$ vanish simultaneously. Counterexample~\ref{def:counterexample} shows that the single-field continuity threshold of~\cite{bakhtin15,balazs11} involving only $-v_{\sigma_0}'(x_0)$ and $\lambda_{\sigma_0}$ then becomes incorrect. Our first contribution is a new continuity threshold at points where multiple vector fields vanish simultaneously, which emerges from the interplay of all vanishing vector fields with the jump rates. 

Our second contribution is a method shift. To study the regularity of the invariant densities on intervals where no $v_\sigma$ vanishes, called noncritical intervals, we resolutely turn to distribution theory and the Kolmogorov forward equation. This equation, which is satisfied by the $p_\sigma$, is a system of first-order differential equations whose first-order coefficients are
the $v_\sigma$. The system thus
degenerates precisely at their zeros, explaining why singularities can only form there. It also yields the optimal regularity assumption: on noncritical intervals, the $v_\sigma$ need only be $C^r$ for the $p_\sigma$ to be $C^r$. Previous results~\cite{bakhtin15} required $C^{r+1}$ vector fields, the additional derivative being an artifact of the method rather than a feature of the problem.

A further advantage of this
formulation on noncritical intervals is that it encompasses, at no additional cost, the position-dependent jump rates~\cite{calvez14,singh20,jain24}
and resetting mechanisms~\cite{evans18,santra20,bressloff20} that arise in applications but lie outside the scope
of~\cite{bakhtin15,balazs11}. The jumps enter the differential equations only as zeroth-order
coefficients and as source terms. Thus, it suffices that the jump rates and the source contributed by the 
resetting be $C^{r-1}$ for the invariant density to be $C^r$. For continuity of the $p_\sigma$ alone, it is enough 
that the jump rates be measurable and the resetting source have a density. 

As an application, we completely describe the shape transition of two RTP models: three velocities in the power-law
potential $\frac{a}{p+1}|x|^{p+1}$ and six velocities in the harmonic potential. Both were
identified as natural next cases in~\cite{basu20} but their shape transition previously seemed out of reach because their invariant measure cannot be computed explicitly. Our approach shows that shape transition is determined by local data only and thus independent of exact solvability. \\

The article is organized as follows. Section~\ref{sec:mathematical_setup} introduces the one-dimensional PDMP framework and exemplifies how it can be used to model RTPs. Section~\ref{sec:main_results} presents the main results. Section~\ref{sec:regularity_noncritical_intervals} concerns the regularity of the invariant measure's density on intervals where no vector field vanishes. Section~\ref{sec:continuity_critical_point} examines the continuity of the density at points where multiple vector fields vanish. Finally, Section~\ref{sec:applications} analyzes the shape transitions of the three-velocity power-law-potential and the six-velocity harmonic-potential models.

\section{Run-and-tumble particles as PDMPs}\label{sec:mathematical_setup}

Dynamical systems with random switching~\cite{bakhtin12,bakhtin15} and, more generally, piecewise-deterministic Markov processes (PDMPs)~\cite{davis93,malrieu16} combine deterministic motion and discrete random jumps. They arise naturally in a wide range of applications, including neuroscience~\cite{pakdaman10,duarte19}, biology~\cite{zeiser08,rudnicki17,lin18} and sampling~\cite{bouchardcote18,fearnhead18,michel20}. However, they have received less attention than classical diffusion processes in the mathematical literature. They also constitute the natural mathematical model for run-and-tumble particles (RTPs)~\cite{hahn23}, which alternate between constant-velocity motion (deterministic dynamics) and stochastic reorientation (random velocity jumps). We start by introducing the subclass of one-dimensional PDMPs that will be used throughout this article.

\begin{Def}[Local characteristics]\label{def:local_characteristics}
Consider a finite index set $\Sigma$ as well as
\begin{itemize}
	\item a family $(v_\sigma)_{\sigma \in \Sigma}$ of locally Lipschitz vector fields $v_\sigma : \mathbb R \to \mathbb R$ s.t.~for all $y_0 \in \mathbb R$ the ODE
	\[
	\partial_t y(t) = v_\sigma(y(t)) \text{ with initial condition } y(0) = y_0
	\]
	can be solved for $t \ge 0$,
	\item a family $(\lambda_{\sigma\tilde\sigma})_{\sigma, \tilde\sigma \in \Sigma}$ of bounded measurable functions $\lambda_{\sigma\tilde\sigma} : \mathbb R \to \mathbb R$ such that \mbox{$\sum_{\tilde \sigma \in \Sigma} \lambda_{\sigma\tilde\sigma} = 0$} for all $\sigma \in \Sigma$ and $\lambda_{\sigma\tilde\sigma} \ge 0$ for all $\sigma \ne \tilde\sigma$,
	\item a family $(\lambda\R_\sigma)_{\sigma\in\Sigma}$ of bounded measurable functions $\lambda\R_\sigma : \mathbb R \to \mathbb R_+$,
	\item a family $\left(Q\R_{(x, \sigma)}\right)_{(x, \sigma) \in \mathbb R \times \Sigma}$ of probability measures on $\mathbb R \times \Sigma$ such that $(x, \sigma) \mapsto Q\R_{(x, \sigma)}(A)$ is measurable for all measurable sets $A$.
\end{itemize}
Also define the total jump rate $\lambda_\sigma := \sum_{\tilde\sigma \ne \sigma} \lambda_{\sigma\tilde\sigma} = -\lambda_{\sigma\sigma}$ and let $(\phi^\sigma_t)$ be the flow induced by $v_\sigma$.
\end{Def}

We are interested in the stochastic process $X_t = (x_t, \sigma_t)$ taking its values in $E = \mathbb R \times \Sigma$, where $x_t$ follows the differential equation
\begin{equation*}
\partial_t x = v_{\sigma_t}(x)
\end{equation*}
and with rate $\lambda_{\sigma_t}(x_t)$ the index $\sigma_t$ jumps to a new state distributed according to
\[
	\sum_{\tilde\sigma \ne \sigma_t} \frac{\lambda_{\sigma_t\tilde \sigma}(x_t)}{\lambda_{\sigma_t}(x_t)} \delta_{\tilde\sigma}.
\]

Furthermore, the couple $(x_t, \sigma_t)$ simultaneously jumps to a new position distributed according to $Q\R_{(x_t, \sigma_t)}$ with rate $\lambda\R_{\sigma_t}(x_t)$. In the context of RTPs, we think of $x$ as a position and of $\sigma$ as a velocity state. In this setting, the first kind of jump corresponds to a jump of the particle's velocity, while the second kind of jump is a position resetting with possible velocity randomization. The construction of the process is made precise by the following definition.

\begin{Def}[One-dimensional piecewise-deterministic Markov process]\label{def:1d_pdmp}
	Let the initial state $(x, \sigma) \in E$ be given and set $(\theta_0, \xi_0, \varsigma_0) = (0, x, \sigma)$. For $n \ge 0$, recursively define the sequence of random variables $(\theta_n, \xi_n, \varsigma_n)_{n \in \mathbb N}$ as follows
	\begin{itemize}
		\item $\theta_{n+1}$ has survivor function
		\begin{equation}\label{eq:survivor_function}
			\mathbb P(\theta_{n+1} > t) = \exp\left( -\int_0^t (\lambda_{\varsigma_n}+ \lambda\R_{\varsigma_n})(\phi^{\varsigma_n}_s(\xi_n)) ds \right),
		\end{equation}
		\item the couple $(\xi_{n+1}, \varsigma_{n+1})$ has distribution
		\[
			\frac{\lambda\R_{\varsigma_n}(\Xi_n) Q\R_{(\Xi_n,{\varsigma_n}) }+ \sum_{\tilde\sigma\ne{\varsigma_n}} \lambda_{{\varsigma_n}\tilde\sigma}(\Xi_n) \delta_{(\Xi_n, \tilde{\sigma})}}{\lambda\R_{\varsigma_{n}}(\Xi_n) +  \lambda_{{\varsigma_n}}(\Xi_n)} \text{ with } \Xi_n = \phi^{\varsigma_n}_{\theta_{n+1}}(\xi_n),
		\]
	\end{itemize}
	and $(\theta_{n+1}, \xi_{n+1}, \varsigma_{n+1})$ is conditionally independent of $(\theta_k, \xi_k, \varsigma_k)_{k \le n - 1}$ and $\theta_n$ given $(\xi_n, \varsigma_n)$. Finally set $T_n = \sum_{k = 0}^n \theta_k$ and
	\[
		X_t = (\phi^{\varsigma_n}_{t - T_n}(\xi_n), \varsigma_n) \text{ for } t \in [T_n, T_{n+1}).
	\]
\end{Def}

The generator~\cite[Def.~14.15]{davis93} constitutes the main computational tool in the study of Markov processes. In particular, it can be used to compute their invariant measure. The following proposition, which directly follows from~\cite[Th.~26.14]{davis93} and~\cite[Th.~25.5]{davis93}, characterizes the generator of 1D PDMPs.

\begin{Prop}[Extended generator]\label{prop:extended_generator} The process $X_t$ is a homogeneous strong Markov process. A bounded measurable function $f : E \to \mathbb R$ is in the domain $D(\mathcal L)$ of its extended generator $\mathcal L$ if and only if
\[
t \mapsto f(\phi^\sigma_t(x), \sigma) \text{ is absolutely continuous (see~\cite[Sec.~11.3]{davis93}) for all } (x, \sigma) \in E
\]
and in that case
\begin{equation}\label{eq:generator}
\mathcal L f(x, \sigma) = \underbrace{v_\sigma(x) \partial_x f(x, \sigma)}_\text{\normalfont determ.~motion} + \underbrace{\sum_{\tilde \sigma \in \Sigma} \lambda_{\sigma\tilde\sigma}(x) f(x, \tilde \sigma)}_\text{\normalfont jumps in $\sigma$ only} + \underbrace{\lambda\R_\sigma(x) \left( Q\R_{(x, \sigma)}(f) - f(x, \sigma) \right)}_\text{\normalfont joint jumps in $x$ and $\sigma$}. 
\end{equation}

\end{Prop}

\begin{Rem}
	We use the convention $\lambda_{\sigma\sigma} = -\sum_{\tilde \sigma \ne \sigma} \lambda_{\sigma \tilde\sigma}$ (see Definition~\ref{def:local_characteristics}). Therefore, in~\eqref{eq:generator}, the $\sigma$-only jumps yield the term $\sum_{\tilde\sigma \in \Sigma} \lambda_{\sigma\tilde\sigma}(x)f(x, \tilde \sigma)$ and {not} $\sum_{\tilde\sigma\ne\sigma} \lambda_{\sigma\tilde\sigma}(x)(f(x, \tilde \sigma)-f(x, \sigma))$.
\end{Rem}

\begin{Rem}
	Unlike~\cite{davis93}, we do not assume that $Q\R_{(x, \sigma)}(\{(x, \sigma)\}) = 0$ for all $(x, \sigma) \in E$. However, our framework can be reconciled with that of~\cite{davis93} by using a construction analogous to the minimal process defined in~\cite[Sec.~4]{durmus21}.
\end{Rem}

We will sometimes make the following simplifying assumption. 

\begin{Ass}\label{ass:constant_jump_rates_no_resetting}
	The $\lambda_{\sigma\tilde\sigma}$ are constant and irreducible and $\lambda\R_\sigma = 0$ (i.e.~no resetting).
\end{Ass}

We conclude this section by exemplifying how PDMPs can be used to model RTPs~\cite{hahn23}. The main example is given by two RTPs interacting through an attractive potential $V$~\cite{ledoussal21,hahn24}. The particles are described by their positions $x_1, x_2 \in \mathbb R$ and velocity states $\sigma_1, \sigma_2 \in \mathbb R$. The positions follow the ODEs
\[
	\partial_t x_1 = f(x_1 - x_2) + v \sigma_1, \qquad \partial_t x_2 = f(x_2 - x_1) + v \sigma_2,
\]
where $f = -V'$ is the inter-particle force and satisfies $f(-x) = -f(x)$. The velocities $\sigma_1, \sigma_2$ are independent Markov jump processes. In the case of bacteria and algae modeled by run-and-tumble particles, the reorientation occurs on a significantly shorter timescale than their directed runs~\cite{berg72}. Therefore, reorientation is often treated as instantaneous, resulting in the transition rates of Figure~\ref{fig:instantaneous_tumble_single_particle_velocity_transition_rates} for the $\sigma_i$. More refined models~\cite{slowman17,guillin24,hahn24,sun25} incorporate an additional $0$-velocity state to account for the non-motile phase during reorientation, leading to the rates of Figure~\ref{fig:finite_tumble_single_particle_velocity_transition_rates}.

\begin{figure}[!tb]
	\centering
	\newcommand{\MyNodeDistance}{1.6cm}
	\newcommand{\MyNodeSize}{0.9cm}
	\newcommand{\MyFont}{\footnotesize}
	\begin{subfigure}[t]{0.45\linewidth}
		\centering
		\begin{tikzpicture}[
		>={Stealth[length=3mm, width=2mm]},
		every edge/.style={thick, draw=blue!70!black},
		node distance=\MyNodeDistance,
		main/.style={
			draw=black!80,
			fill=blue!10,
			circle,
			minimum size=\MyNodeSize,
			font=\MyFont,
			drop shadow,
			minimum size=.8cm
		}
		]
		\node[main] (I+) {$-1$};
		\node[main] (I-) [right=of I+] {$+1$};
		\path[->]
		(I+) edge[bend left=30] node[fill=white, inner sep=2pt] {$\omega$} (I-)
		(I-) edge[bend left=30] node[fill=white, inner sep=2pt] {$\omega$} (I+);
		\end{tikzpicture}
		\caption{Instantaneous tumble}
		\label{fig:instantaneous_tumble_single_particle_velocity_transition_rates}
	\end{subfigure}
	\hspace{0.02\linewidth}
	\begin{subfigure}[t]{0.45\linewidth}
		\centering
		\begin{tikzpicture}[
		>={Stealth[length=3mm, width=2mm]},
		every edge/.style={thick, draw=red!70!black},
		node distance=\MyNodeDistance,
		main/.style={
			draw=black!80,
			fill=red!10,
			circle,
			minimum size=\MyNodeSize,
			font=\MyFont,
			drop shadow,
			minimum size=.8cm
		}
		]
		\node[main] (F+) {$-1$};
		\node[main] (F0) [right=of F+] {$0$};
		\node[main] (F-) [right=of F0] {$+1$};
		\path[->]
		(F+) edge[bend left=25] node[fill=white, inner sep=2pt] {$\alpha$} (F0)
		(F0) edge[bend left=25] node[fill=white, inner sep=2pt] {$\tfrac{\beta}{2}$} (F+)
		(F0) edge[bend right=25] node[fill=white, inner sep=2pt] {$\tfrac{\beta}{2}$} (F-)
		(F-) edge[bend right=25] node[fill=white, inner sep=2pt] {$\alpha$} (F0);
		\end{tikzpicture}
		\caption{Finite tumble}
		\label{fig:finite_tumble_single_particle_velocity_transition_rates}
	\end{subfigure}
	\caption{Markov jump process followed by the single-particle velocity states}
	\label{fig:single_particle_velocity_transition_rates}
\end{figure}

The process $(x_1, x_2, \sigma_1, \sigma_2)$ does not reach a steady state so the particle separation $x = x_2 - x_1$ and relative velocity state $\sigma = \sigma_2 - \sigma_1$ are considered instead. The particle separation $x$ obeys the differential equation
\[
	\partial_t x = 2f(x) + v\sigma.
\]
If the $\sigma_i$ follow Figure~\ref{fig:instantaneous_tumble_single_particle_velocity_transition_rates} then $\sigma = \sigma_2 - \sigma_1$ is a Markov jump process following Figure~\ref{fig:instantaneous_tumble_relative_particle_velocity_transition_rates}. However, if the $\sigma_i$ follow Figure~\ref{fig:finite_tumble_single_particle_velocity_transition_rates} then $\sigma_2 - \sigma_1$ no longer has the Markov property. This can be fixed by splitting $\sigma_2 - \sigma_1 = 0$ into two different states $0_\pm$ and $0_0$ corresponding to $\sigma_1 = \sigma_2 = \pm 1$ and $\sigma_1 = \sigma_2 = 0$ respectively. The resulting transition rates for $\sigma$ are shown in Figure~\ref{fig:finite_tumble_relative_particle_velocity_transition_rates}.

\begin{figure}[!tb]
	\centering

	\newcommand{\ND}{1.1cm}
	\newcommand{\NS}{0.6cm}
	\newcommand{\FS}{\footnotesize}
	\newcommand{\TipLen}{2mm}
	\newcommand{\TipWid}{1.5mm}
	
	\begin{subfigure}[t]{0.38\linewidth}
		\centering
		\begin{tikzpicture}[
		>={Stealth[length=\TipLen, width=\TipWid]},
		every edge/.style={thick, draw=blue!70!black},
		node distance=\ND,
		main/.style={
			draw=black!80,
			fill=blue!10,
			circle,
			minimum size=\NS,
			font=\FS,
			drop shadow,
			minimum size=.8cm
		}
		]
		\node[main] (m2) {$-2$};
		\node[main] (zero) [right=of m2] {$0$};
		\node[main] (p2) [right=of zero] {$+2$};
		
		\path[->]
		(zero) edge[bend right=25] node[fill=white, inner sep=1.5pt] {$\omega$} (m2)
		(zero) edge[bend left=25]  node[fill=white, inner sep=1.5pt] {$\omega$} (p2)
		(m2)   edge[bend right=25] node[fill=white, inner sep=1.5pt] {$2\omega$} (zero)
		(p2)   edge[bend left=25]  node[fill=white, inner sep=1.5pt] {$2\omega$} (zero);
		\end{tikzpicture}
		\caption{Instantaneous tumble}
		\label{fig:instantaneous_tumble_relative_particle_velocity_transition_rates}
	\end{subfigure}
	\begin{subfigure}[t]{0.58\linewidth}
		\centering
		\begin{tikzpicture}[
		>={Stealth[length=\TipLen, width=\TipWid]},
		every edge/.style={thick, draw=red!70!black},
		node distance=\ND,
		main/.style={
			draw=black!80,
			fill=red!10,
			circle,
			minimum size=\NS,
			font=\FS,
			drop shadow,
			minimum size=.8cm
		}
		]

		\node[main] (m2)  {$-2$};
		\node[main] (m1)  [right=of m2]    {$-1$};
		\node[main] (0pm) [above right=of m1] {$0_\pm$};
		\node[main] (00)  [below right=of m1] {$0_0$};
		\node[main] (p1)  [below right=of 0pm] {$+1$};
		\node[main] (p2)  [right=of p1]       {$+2$};
		
		\path[->]
		(m2) edge[bend left=25]  node[fill=white, inner sep=1.5pt] {$2\alpha$} (m1)
		(m1) edge[bend left=25] node[fill=white, inner sep=1.5pt] {$\frac\beta2$}  (m2)
		(p2) edge[bend right=25] node[fill=white, inner sep=1.5pt] {$2\alpha$} (p1)
		(p1) edge[bend right=25]  node[fill=white, inner sep=1.5pt] {$\frac\beta2$}  (p2)
		
		(m1) edge[bend left=25]  node[fill=white, inner sep=1.5pt] {$\tfrac\beta2$} (0pm)
		(m1) edge[bend right=25] node[fill=white, inner sep=1.5pt] {$\alpha$}      (00)
		(p1) edge[bend right=25] node[fill=white, inner sep=1.5pt] {$\tfrac\beta2$} (0pm)
		(p1) edge[bend left=25]  node[fill=white, inner sep=1.5pt] {$\alpha$}      (00)
		
		(0pm) edge[bend left=25]  node[fill=white, inner sep=1.5pt] {$\alpha$}      (m1)
		(0pm) edge[bend right=25] node[fill=white, inner sep=1.5pt] {$\alpha$}      (p1)
		(00)  edge[bend right=25] node[fill=white, inner sep=1.5pt] {$\beta$} (m1)
		(00)  edge[bend left=25]  node[fill=white, inner sep=1.5pt] {$\beta$} (p1)
		;
		\end{tikzpicture}
		\caption{Finite tumble}
		\label{fig:finite_tumble_relative_particle_velocity_transition_rates}
	\end{subfigure}
	
	\caption{Markov jump process followed by the relative velocity state}
	\label{fig:relative_particle_velocity_transition_rates}
\end{figure}

The following two processes will be studied in detail in the present article
\begin{itemize}
	\item the \emph{two-particle instantaneous power-law process} where $f(x) = -a (\sgn x) |x|^p$ with $a > 0$, $p > 1$ and $\sigma$ follows the rates of Figure~\ref{fig:instantaneous_tumble_relative_particle_velocity_transition_rates} (equivalently $V(x) = \frac{a}{p+1}|x|^{p+1}$ and $\sigma_1,\sigma_2$ follow  Figure~\ref{fig:instantaneous_tumble_single_particle_velocity_transition_rates}),
	\item the \emph{two-particle finite harmonic process} where $f(x) = -ax$ with $a > 0$ and $\sigma$ follows the rates of Figure~\ref{fig:finite_tumble_relative_particle_velocity_transition_rates} (equivalently $V(x) = \frac{a}{2}x^2$ and $\sigma_1,\sigma_2$ follow  Figure~\ref{fig:finite_tumble_single_particle_velocity_transition_rates}).
\end{itemize}

\begin{Def}[Two-particle instantaneous power-law process]\label{def:instantaneous_powerlaw_process}
	Let $a, v, \omega > 0$ and $p > 1$. We call \emph{two-particle instantaneous power-law process} the 1D PDMP obtained by taking 
	\begin{align*}
	\Sigma &= \{2, 0, -2\}, & v_\sigma(x) &= -2 a (\sgn x) |x|^p + v \sigma, & \lambda\R_\sigma(x) &= 0,
	\end{align*}
	and 
	\[
		(\lambda_{\sigma\tilde\sigma}(x))_{\sigma, \tilde\sigma \in \Sigma} =  \bordermatrix{ & \tilde \sigma = 2 & \tilde \sigma =  0 & \tilde \sigma =  -2 \cr
			\sigma = 2 & -2  {\omega} & 2  {\omega} & 0 \cr
			\sigma = 0 &{\omega} & -2  {\omega} & {\omega} \cr
			\sigma = -2 & 0 & 2  {\omega} & -2  {\omega}},
	\]
	in Definition~\ref{def:1d_pdmp}. The $Q\R_{(x, \sigma)}$ need not be specified as $\lambda\R_\sigma =0$.
\end{Def}

\begin{Def}[Two-particle finite harmonic process]\label{def:finite_harmonic_process}
	Let $a, v, \alpha, \beta> 0$. We call \emph{two-particle finite harmonic process} the 1D PDMP obtained by taking
	\begin{align*}
\Sigma &= \{2, 1, 0_\pm, 0_0, -1, -2\}, & v_\sigma(x) &= -2a x + v \sigma, & \lambda\R_\sigma(x) &= 0,
\end{align*}
with the convention that $0_\pm \cdot v = 0_0 \cdot v = 0$ and
\[
(\lambda_{\sigma\tilde\sigma}(x))_{\sigma, \tilde\sigma \in \Sigma} = 
\bordermatrix{ & \tilde \sigma = 2 & \tilde \sigma = 1 & \tilde \sigma = 0_\pm & \tilde \sigma = 0_0 & \tilde \sigma = -1 & \tilde \sigma = -2 \cr
	\sigma = 2 & -2  {\alpha} & 2  {\alpha} & 0 & 0 & 0 & 0 \cr
	\sigma =1 & \frac{1}{2}  {\beta} & -{\alpha} - {\beta} & \frac{1}{2}  {\beta} & {\alpha} & 0 & 0 \cr
	\sigma = 0_\pm& 0 & {\alpha} & -2  {\alpha} & 0 & {\alpha} & 0 \cr
	\sigma = 0_0 & 0 & {\beta} & 0 & -2  {\beta} & {\beta} & 0 \cr
	\sigma = -1 & 0 & 0 & \frac{1}{2}  {\beta} & {\alpha} & -{\alpha} - {\beta} & \frac{1}{2}  {\beta} \cr
	\sigma = -2 & 0 & 0 & 0 & 0 & 2  {\alpha} & -2  {\alpha} 
},
\]
in Definition~\ref{def:1d_pdmp}. The $Q\R_{(x, \sigma)}$ need not be specified as $\lambda\R_\sigma =0$.
\end{Def}

In the remainder of this article, the two-particle instantaneous power-law process (resp.~two-particle finite harmonic process) will often simply be called \emph{power-law process} (resp.~\emph{harmonic process}). We end this section with an example displaying joint jumps in $x$ and $\sigma$, and therefore not satisfying Assumption~\ref{ass:constant_jump_rates_no_resetting}. It models a single free RTP which resets its position and randomizes its velocity with rate $r > 0$~\cite{evans18}. Here, $x$ denotes the single particle's position and $\sigma$ its velocity state.

\begin{Def}[Single-particle resetting process~\cite{evans18}] Let $v, \omega, r > 0$. We call \emph{single-particle resetting process} the 1D PDMP obtained by taking
\begin{align*}
	\Sigma &= \{1, -1\}, & v_\sigma(x) &= v \sigma, & \lambda\R_\sigma(x) &= r, & Q\R_{(x, \sigma)} &= \delta_0 \otimes \left( \frac12 \delta_{-1} + \frac12 \delta_1 \right),
\end{align*}
and
\[
(\lambda_{\sigma\tilde\sigma}(x))_{\sigma, \tilde\sigma \in \Sigma} = 
\bordermatrix{ & \tilde \sigma = 1 & \tilde \sigma = -1 \cr
	\sigma = 1 & -\omega & \omega \cr
	\sigma = -1 & \omega & -\omega \cr
},
\]
in Definition~\ref{def:1d_pdmp}.
\end{Def}

\section{Main results}\label{sec:main_results}

Subsection~\ref{subsec:main_results_regularity_noncritical_intervals} concerns the regularity of the invariant measure on intervals where no $v_\sigma$ vanishes. Subsection~\ref{subsec:general_approach_shape_transition} reviews the known behavior at a point where a single $v_\sigma$ vanishes and Subsection~\ref{subsec:intro_C0_critical_points} treats points where several vanish at once. Subsection~\ref{subsec:application_shape_transition} applies these results to the shape transition of run-and-tumble particles.

\subsection{Regularity on noncritical intervals}\label{subsec:main_results_regularity_noncritical_intervals}

The only points where the invariant densities can display singularities are those where the vector fields $v_\sigma$ vanish. This is the essence of~\cite[Th.~1]{bakhtin15}.
Our first main contribution, Theorem~\ref{thm:Cr_regularity_noncritical_interval} below, identifies the mechanism governing regularity. The invariant densities solve a first-order linear system in the sense of distributions, which is degenerate exactly where the $v_\sigma$ vanish, since these are the first-order coefficients. Jumps enter only as zeroth-order coefficients and source terms, so they play no role in the degeneracy. As a consequence, Theorem~\ref{thm:Cr_regularity_noncritical_interval} does not require Assumption~\ref{ass:constant_jump_rates_no_resetting}. In other words, it applies to processes with position-dependent jump rates~\cite{calvez14,singh20,jain24} and resetting~\cite{evans18,santra20,bressloff20} not covered by previous results but arising in applications. Note that density singularities appear precisely at the configurations that are the most likely in the low jump rate regime, which underscores the importance of locating them.

\begin{Def} Let $\pi$ be a measure on $\mathbb R \times \Sigma$ and $I \subset \mathbb R$ be an open interval.
\begin{itemize}
	\item  We say that $\pi$ has a density (resp.~is $C^r$) on $I \times \{ \sigma \}$ if there exist $p \in L^1(I)$ (resp.~$p \in C^r(I)$) such that
	\[
	\pi(f) = \int_I f(x, \sigma) p(x) dx
	\]
	for all bounded measurable functions $f : E \to \mathbb R$ vanishing outside of $I \times \{\sigma\}$.
	\item We say that $\pi$ has a density (resp.~is $C^r$) on $I$ if $\pi$ has a density (resp.~is $C^r$) on $I \times \{\sigma\}$ for all $\sigma \in \Sigma$. 
	\item When $\pi$ has a density (resp.~is $C^r$) on $I = \mathbb R$, we simply say that $\pi$ has a density (resp. is~$C^r$). This is equivalent to the existence of $p_\sigma \in L^1(\mathbb R)$ (resp.~$p_\sigma \in C^r(\mathbb R)$) such that
	\[
	\pi = \sum_{\sigma \in \Sigma} p_\sigma(x) dx \otimes \delta_\sigma.
	\]
	In a slight abuse of terminology, we refer to the  $p_\sigma$ as invariant densities.
\end{itemize}	
\end{Def}

The following auxiliary measure captures the regularity of the resetting mechanism.

\begin{Def}\label{def:kappa_pi}
	To every bounded measure $\pi$ on $E$ we associate another measure $\kappa^\pi$ defined by
	\[
	\kappa^\pi(f)= \int_E \lambda\R_\sigma(x) Q\R_{(x, \sigma)}(f) d\pi(x, \sigma)
	\]
	for all bounded measurable $f : E \to \mathbb R$.
\end{Def}

\begin{Thm}\label{thm:Cr_regularity_noncritical_interval}
	Let $\pi$ be an invariant measure, $I$ an open interval and $r \ge 1$. If
	\begin{itemize}
		\item the \( v_\sigma \) do not vanish on $I$ and are \( C^r \) on \( I \),
		\item \( \lambda_{\sigma\tilde\sigma} \) and \( \lambda\R_\sigma \) are \( C^{r-1} \) on $I$,
		\item $\kappa^\pi$ is $C^{r-1}$ on $I$,
	\end{itemize}
	then $\pi$ is $C^r$ on $I$.
\end{Thm}

\begin{Rem}
	One has that $\kappa^\pi$ is $C^{r-1}$ on $I$ for all measures $\pi$ under either of the following conditions
	\begin{itemize}
		\item $\lambda_\sigma\R(x) Q\R_{(x,\sigma)}(I \times \Sigma) = 0$ for all $(x, \sigma) \in E$ (i.e.~no resetting to $I \times \Sigma$),
		\item there exist a finite number of measures $Q\R_1, \ldots, Q\R_N$, all of which are $C^{r-1}$ on $I$, such that $Q\R_{(x, \sigma)} \in \{ Q\R_1, \ldots, Q\R_N \}$ for all $(x, \sigma) \in E$.
	\end{itemize}
\end{Rem}

Note that we obtain an additional derivative compared to~\cite[Th.~1]{bakhtin15}, which requires \( v_\sigma \in C^{r+1}(I) \) to ensure \( p_\sigma \in C^r(I)\). Considering examples such as~\cite[Prop.~3.12]{faggionato09}, where the invariant measure is explicit, shows that Theorem~\ref{thm:Cr_regularity_noncritical_interval} captures the optimal regularity assumption on $v_\sigma$. In other words, the derivative lost in~\cite{bakhtin15} is an artifact of the method rather than a feature of the problem. The next theorem implies the continuity of the invariant densities even on intervals where the $\lambda_{\sigma\tilde\sigma}$ are discontinuous, as in~\cite{fontbona12,calvez14}. This is coherent with the fact that $\lambda_{\sigma\tilde\sigma}$ need only be $C^{r-1}$ in Theorem~\ref{thm:Cr_regularity_noncritical_interval}. Heuristically, this regularizing effect occurs because the jump rates are integrated along the flow of the ODEs, as reflected, e.g., in the survivor function~\eqref{eq:survivor_function}. 

\begin{Thm}\label{thm:continuity_noncritical_interval}
	If the invariant measure $\pi$ has a density on the open interval $I$ and $\sigma_0 \in \Sigma$ is~s.t.
	\begin{itemize}
		\item the vector field $v_{\sigma_0}$ does not vanish on $I$,
		\item $\kappa^\pi$ has a density on $I \times \{\sigma_0\}$,
	\end{itemize}
	then $\pi$ is $C^0$ on $I \times \{\sigma_0\}$.
\end{Thm}

\begin{Rem} Under either of the following conditions, $\kappa^\pi$ has a density on $I \times \{ \sigma_0 \}$ for all measures~$\pi$
	\begin{itemize}
		\item $\lambda\R_\sigma(x)Q\R_{(x,\sigma)}(I \times \{\sigma_0\}) = 0$ for all $(x, \sigma) \in E$ (i.e.~no resetting to $I \times \{ \sigma_0\}$),
		\item there exist a finite number of measures $Q\R_1, \ldots, Q\R_N$, all of which have a density on $I \times \{\sigma_0\}$, such that $Q\R_{(x, \sigma)} \in \{ Q\R_1, \ldots, Q\R_N \}$ for all $(x, \sigma) \in E$.
	\end{itemize}
\end{Rem}

\begin{Rem}
Note that in Theorem~\ref{thm:Cr_regularity_noncritical_interval}, the density of $\pi$ is a consequence, whereas in Theorem~\ref{thm:continuity_noncritical_interval}, it is an assumption. This assumption can, however, be verified a priori using the techniques of~\cite{bakhtin12,benaim15}.
\end{Rem}

The proofs we provide for Theorems~\ref{thm:Cr_regularity_noncritical_interval} and~\ref{thm:continuity_noncritical_interval} are short. They proceed by studying the Kolmogorov forward equation satisfied by the invariant measure with distributional tools, continuing the systematic use of the generator to investigate the invariant measure of PDMPs~\cite{hahn23}.

\subsection{Continuity at simply critical points}\label{subsec:general_approach_shape_transition}

Given that invariant densities can only develop singularities at points where the vector fields $v_\sigma $ vanish, it is natural to investigate the conditions under which such divergences actually occur. We start by reviewing the results of~\cite{bakhtin15,balazs11}, established under Assumption~\ref{ass:constant_jump_rates_no_resetting} as well as Assumption~\ref{ass:structure_assumption_single_vanishing_vector_fields} below, which are both assumed to hold for the rest of Subsection~\ref{subsec:general_approach_shape_transition}. We conclude with Counterexample~\ref{def:counterexample}, which shows that this picture breaks down as soon as several vector fields vanish at the same point.
\begin{Ass}\label{ass:structure_assumption_single_vanishing_vector_fields} 
The state $(x_0, \sigma_0) \in E$ is such that $v_{\sigma_0}(x_0) = 0$ and $v_\sigma(x_0) \ne 0$ for $\sigma \ne \sigma_0$.
\end{Ass}

Heuristically, if $v_{\sigma_0}(x_0) = 0$ and $v_{\sigma_0}'(x_0) < 0$ then the deterministic dynamics locally converges to $x_0$ at exponential speed, leading to a stark accumulation of mass at that point for the invariant measure. This is balanced by the stochastic jumps, which change the vector field with rate $\lambda_{\sigma_0}$, stopping the deterministic contraction. The presence or absence of a singularity at $x_0$ is determined by this competition between the contracting deterministic dynamics and the stochastic jumps. By~\cite[Th.~3]{bakhtin15} one has that $p_{\sigma_0}$
\begin{itemize}
	\item diverges at $x_0$ if $v_{\sigma_0}'(x_0) < 0$ and $\lambda_{\sigma_0} < -v_{\sigma_0}'(x_0)$,
	\item is continuous at $x_0$ if $v_{\sigma_0}'(x_0) < 0$ and $\lambda_{\sigma_0} > -v_{\sigma_0}'(x_0)$,
	\item is continuous at $x_0$ if $v_{\sigma_0}'(x_0) > 0$, regardless of the jump rate $\lambda_{\sigma_0}$.
\end{itemize}
The density $p_\sigma$ is continuous at $x_0$ for $\sigma \ne \sigma_0$ by Theorem~\ref{thm:continuity_noncritical_interval}. In fact, \cite[Th.~3]{bakhtin15} gives the precise asymptotic behavior of $p_{\sigma_0}$ at $x_0$. The result when $v_{\sigma_0}'(x_0) > 0$ is intuitive because, in that case, $x_0$ is repelling instead of attracting, so there is no accumulation of probability mass at that point. It follows from~\cite[Th.~1]{balazs11} that if
\[
	|v_{\sigma_0}(x)| \sim C|x - x_0|^\nu \text{ for } C >0, \nu > 1 \text{ when } x \to x_0
\]
then the density is locally bounded at $x_0$ irrespective of $\lambda_{\sigma_0}$. In fact, one even expects the density to be continuous in this case. Indeed, in this setting, the convergence to $x_0$ is sub-exponential and thus cannot compete with the exponentially distributed stochastic jumps. Although this falls outside the scope of the present article, one expects that if $\nu < 1$ and $x_0$ is attracting, then the deterministic dynamics reaches $x_0$ in finite time, causing the invariant measure to have an atom at $(x_0, \sigma_0)$.

Notably, in the case of the two-particle finite harmonic process, $v_{0_\pm}(0) = v_{0_0}(0) = 0$ so Assumption~\ref{ass:structure_assumption_single_vanishing_vector_fields} is not satisfied at $x = 0$ and continuity at that point cannot be obtained from~\cite{bakhtin15,balazs11}. In fact, the following counterexample shows that the picture becomes more involved when multiple vector fields vanish at the same time.

\begin{Cex}\label{def:counterexample}
	Consider the 1D PDMP obtained by taking $\Sigma = \{1, 2, 3\}$ as well as
	\begin{align*}
	v_1(x) &= -x, & v_2(x) &= -2x(1 - x), & v_3(x) &= 1 - x,
	\end{align*}
	and
	\begin{align*}
	(\lambda_{\sigma\tilde\sigma}(x))_{\sigma, \tilde\sigma \in \Sigma} &=  \bordermatrix{ & \tilde \sigma = 1 & \tilde \sigma =  2 & \tilde \sigma =  3 \cr
		\sigma = 1 & -2  {\omega} & 2  {\omega} & 0 \cr
		\sigma = 2 &{\omega} & -2  {\omega} & {\omega} \cr
		\sigma = 3 & 0 & 2  {\omega} & -2  {\omega}}, & \lambda\R_\sigma(x) &= 0,
	\end{align*}
	in Definition~\ref{def:1d_pdmp}. The $Q\R_{(x, \sigma)}$ need not be specified as $\lambda\R_\sigma =0$.
\end{Cex}

Explicitly solving the Kolmogorov forward equation (see Lemma~\ref{lem:fokker_planck}) shows that the invariant measure of this process is unique and has the form $\pi = \frac1Z\sum_{\sigma \in \Sigma} p_\sigma(x) dx \otimes \delta_\sigma$ where
\begin{align*}
p_1(x) &= 1_{\{0 < x < 1\}} \cdot 2 x^{\frac{3 - \sqrt{5}}{2}\omega -1} (1-x)^{\frac{1+\sqrt{5}}{2} \omega }, \\
p_2(x) &= 1_{\{0 < x < 1\}} \cdot (1 + \sqrt{5}) x^{\frac{3 - \sqrt{5}}{2}\omega -1} (1-x)^{\frac{1+\sqrt{5}}{2} \omega -1 }, \\
p_3(x) &= 1_{\{0 < x < 1\}} \cdot (4+2\sqrt{5}) x^{\frac{3 - \sqrt{5}}{2}\omega} (1-x)^{\frac{1+\sqrt{5}}{2} \omega -1 },
\end{align*}
and $Z > 0$ is a normalizing constant. In particular
\[
p_1 \text{ and } p_2 \text{ are continuous at } 0 \iff \omega > \frac{3 + \sqrt{5}}{2}.
\]
Using~\cite[Th.~2]{bakhtin15} even though it cannot be applied because Assumption~\ref{ass:structure_assumption_single_vanishing_vector_fields} is not satisfied would yield
\begin{align*}
p_1 \text{ is continuous at } 0 &\iff \omega> 1/2, & p_2 \text{ is continuous at } 0 &\iff \omega > 1.
\end{align*}
Thus, while the invariant densities remain continuous above a critical jump-rate threshold and diverge below, this threshold is higher than anticipated. In other words, at multiply critical points the single-field threshold of~\cite{bakhtin15} predicts an incorrect value, indicating that a genuinely different mechanism is at play. Understanding this discrepancy and, more generally, what happens when multiple vector fields vanish simultaneously is the topic of the next subsection.

\subsection{Continuity at multiply critical points}\label{subsec:intro_C0_critical_points}

When a single vector field vanishes, jumping to any other vector field stops the deterministic contraction. However, when multiple vector fields have a common zero, the process can jump from one vanishing $v_\sigma$ to another. The deterministic contraction then continues, possibly with a different rate. This suggests that understanding singularity formation requires analyzing the combined contraction effect of all vanishing $v_\sigma$ and the overall rate at which the system exits this group of vector fields. It also suggests that, if Assumption~\ref{ass:structure_assumption_single_vanishing_vector_fields} is not satisfied but direct jumps between vanishing vector fields are not possible, the continuity threshold of~\cite{bakhtin15} should remain valid. This applies, in particular, to the harmonic process. We now investigate the continuity of the subset of invariant densities $(p_\sigma)_{\sigma \in S}$ at $x_0$ under the following assumption on the index set $S$.

\begin{Ass}\label{ass:structure_assumption_multiple_vanishing_vector_fields}  One has that $\Sigma_0 := \{ \sigma \in \Sigma : v_\sigma(x_0) = 0 \} \ne \Sigma$  and $S \subset \Sigma_0$ is non-empty.
\end{Ass}

The case $\Sigma_0 = \Sigma$ is considered in~\cite{benaim19}, although with a focus on whether or not the invariant measure has a density rather than its regularity. We work under Assumptions~\ref{ass:constant_jump_rates_no_resetting} and~\ref{ass:structure_assumption_multiple_vanishing_vector_fields} and assume without loss of generality that $x_0 = 0$ during the rest of Subsection~\ref{subsec:intro_C0_critical_points}. It is useful to study the continuity of $(p_\sigma)_{\sigma \in S}$ separately for different index sets $S \subset \Sigma_0$.

Observe that the condition $\lambda_{\sigma_0} > -v_{\sigma_0}'(0)$ ensuring continuity in~\cite[Th.~3]{bakhtin15} can be rewritten as
\[
	\mathbb E_{\sigma_0} \left[ \int_0^{\tilde\tau} e^{-t v_{\sigma_0}'(0)} dt \right] < +\infty \text{ where } \tilde \tau = \inf\{ t \ge 0 : \sigma_t \ne \sigma_0 \}.
\]
In the case of a subset of vanishing vector fields $(v_\sigma)_{\sigma \in S}$, this criterion generalizes to
\begin{equation}\label{eq:Ea_finite}
\mathbb E_\sigma \left[ \int_0^\tau e^{-\int_0^t v_{\sigma_s}'(0) ds} dt \right] < +\infty \text{ where } \tau = \inf \{ t \ge 0 : \sigma_t \notin S \},
\end{equation}
thus giving a way to compare the rate at which the process leaves this subset and the joint contraction rate. This paper's second main result is that, in essence, the $(p_\sigma)_{\sigma \in S}$ are continuous when~\eqref{eq:Ea_finite} is satisfied and diverge when it is not. Importantly, the expectation in~\eqref{eq:Ea_finite} can be computed explicitly by solving a system of linear equations (see Lemma~\ref{lem:Ec_ec_general_properties}). This makes~\eqref{eq:Ea_finite} effective in the sense that it can easily be checked on explicit models. In particular, in the case of Counterexample~\ref{def:counterexample}, taking $S = \{1, 2\}$, one has (see Proposition~\ref{prop:counterexample_continuity_threshold})
\begin{align*}
\eqref{eq:Ea_finite} \iff \omega > \frac{3+ \sqrt{5}}{2} \text{ for } \sigma = 1, 2,
\end{align*}
thus recovering the correct continuity threshold. All of this is made rigorous in the upcoming Theorem~\ref{thm:continuity_critical_point} by considering the following expectations, which are a slight generalization of~\eqref{eq:Ea_finite}.

\begin{Def} For all families of reals $c = (c_\sigma)_{\sigma \in S}$ define
	\begin{align*}
	E^{c}_\sigma := \mathbb E_\sigma \left[ \int_0^\tau e^{\int_0^t c_{\sigma_s}ds} dt \right]
	\end{align*}
	where $\tau = \inf \{ t \ge 0 : \sigma_t \notin S \}$.
\end{Def}
The key idea is to reformulate continuity as the integrability of certain functions, as follows:
\begin{itemize}
	\item If $I_\mathrm{d}(\epsilon, \eta) := \sum_{\sigma \in S} \int_0^\epsilon x^{-1+\eta} p_\sigma(x) dx = +\infty$ then $\varlimsup_{x \to 0+} \sum_{\sigma \in S} p_\sigma(x) = +\infty$ so $\sum_{\sigma \in S} p_\sigma$ cannot be continuous at $0$.
	\item If $I_\mathrm{c}(\epsilon) := \sum_{\sigma \in S} \int_0^\epsilon \frac{1}{x (\log x)^2} p_\sigma(x) dx < +\infty$ and 
	\[
	p_\sigma(x) \sim C x^\nu (\log x)^k \text{ when } x \to \text{$0+$} \text{ with } C \ne 0, \nu \in \mathbb R \text{ and } k \in \mathbb N,
	\]
	then $\nu > 0$ or $\nu = k = 0$. Hence $p_\sigma$ admits a limit to the right at $0$.
\end{itemize}
To analyze $I_\mathrm{d}(\epsilon, \eta)$ and $I_\mathrm{c}(\epsilon)$, we relate them to $E^c_\sigma$ by linearizing the deterministic dynamics around $0$. This enables the estimation of both integrals. For analytic vector fields, the asymptotic behavior $p_\sigma(x) \sim C x^\nu (\log x)^k$ can be shown using the theory of differential equations with regular singular points~\cite[Sec.~3.11]{taylor21}. Note that linearizing the deterministic dynamics was also the key to the nature of the invariant measure in~\cite{benaim19}.

\begin{Def}\label{def:properties_of_S} Define $S_\mathrm{in} = \{ \sigma \in \Sigma \setminus S : \max_{\tilde\sigma \in S} \lambda_{\sigma\tilde\sigma} > 0\}$. We say that $S$ is
	\begin{itemize}
		\item \emph{backward-complete} if $S_\mathrm{in} \cap \Sigma_0 = \emptyset$,
		\item \emph{irreducible} if for all $\sigma, \tilde\sigma \in S$ there exists a sequence $\sigma = \sigma_1, \sigma_2, \ldots, \sigma_N = \tilde\sigma \in S$ such that $\lambda_{\sigma_n\sigma_{n+1}} > 0$ for $n = 1, 2, \ldots, N - 1$.
	\end{itemize}
\end{Def}

In words, backward-completeness of the index set $S$ means that any state in $\Sigma_0$ leading to $S$ must already be contained in $S$.

\begin{Ass}\label{ass:integrability} One has that:

\begin{enumerate}[label=(D\arabic*)]
	\item \label{ass:D1_compact_set} There exists a compact set $K \subset \mathbb R$ such that $0 \in \mathring{K}$ and
	\[
		\phi^\sigma_t(K) \subset K \text{ for all } t \ge 0 \text{ and } \sigma \in \Sigma.
	\]
	\item \label{ass:D2_right_diff} For all $\sigma \in S$
	\[
	v_\sigma(x) = -a_\sigma x + o(x) \text{ when } x \to 0+
	\]
	where $a_\sigma> 0$.
	\item \label{ass:D3_S_irreducible} The index set $S$ is irreducible.
	\item \label{ass:D4_TV_convergence} The invariant measure $\pi = \sum_{\sigma \in \Sigma} p_\sigma(x)dx\otimes \delta_\sigma$ is unique, admits a density and satisfies
	\[
\sup_{(x, \sigma) \in K \times \Sigma} \Vert \delta_{(x, \sigma)} P_t - \pi \Vert_\mathrm{TV} \to 0 \text{ when } t \to +\infty,
\]
where the total variation distance $\Vert \cdot\Vert_\mathrm{TV}$ is defined by $\Vert \nu - \mu \Vert_\mathrm{TV} = \sup \nu(A) - \mu(A)$ with the $\sup$ running over all measurable sets $A$. Moreover, $\delta_{(x, \sigma)} P_t$ is the law at time $t$ of the process started from the initial state $(x, \sigma)$.

	\item \label{ass:D5_inside_support} The invariant measure satisfies $\pi([0, \epsilon] \times S) > 0$ for all $\epsilon> 0$.
\end{enumerate}
\end{Ass}

\begin{Rem} Assumptions~\ref{ass:D4_TV_convergence} and~\ref{ass:D5_inside_support} can be checked using~\cite[Cor.~2.7]{benaim18} and~\cite[Sec.~6]{bakhtin15} respectively.
\end{Rem}

\begin{Ass}\label{ass:continuity}
	One has:
\begin{enumerate}[label=(E\arabic*)]
		\item \label{ass:E1_analytic} The $v_\sigma$ are all analytic at $x = 0$ and $a_\sigma := -v'_\sigma(0) \ne 0$ for all $\sigma \in \Sigma_0$.
		\item \label{ass:E2_spectral} The matrix $B_0 = ((B_0)_{\sigma\tilde\sigma})_{\sigma, \tilde\sigma \in \Sigma}$ defined by
		\[
		(B_0)_{\sigma\tilde\sigma} = \left\{
		\begin{array}{cl}
		-\lambda_{\tilde\sigma\sigma}/a_{\tilde\sigma} &\text{ if } \tilde\sigma \in \Sigma_0, \\
		0 &\text{ if } \tilde\sigma \notin \Sigma_0,
		\end{array}
		\right.
		\]
		is diagonalizable and all its eigenvalues are real.
		\item \label{ass:E3_invertibility} The matrix $A = (A_{\sigma\tilde\sigma})_{\sigma, \tilde\sigma\in S \cup S_\mathrm{in}}$ defined by
		\[
		A_{\sigma\tilde\sigma} = \left\{
		\begin{array}{cl}
		\lambda_{\tilde\sigma\sigma} + a_\sigma 1_{\{\sigma = \tilde\sigma\}} & \text{ if } \sigma \in S,\\
		1_{\{\sigma = \tilde\sigma\}} &\text{ if } \sigma \in S_\mathrm{in}, \\
		\end{array}
		\right.
		\]
		is invertible.
	\end{enumerate}
\end{Ass}

\begin{Thm}[Continuity at critical point]\label{thm:continuity_critical_point} Assume that Assumptions~\ref{ass:constant_jump_rates_no_resetting}, \ref{ass:structure_assumption_multiple_vanishing_vector_fields} and~\ref{ass:integrability} are satisfied.
\begin{enumerate}[label=(\roman*)]
	\item If there exists $\gamma > 0$ such that
	\begin{equation}\label{eq:Ec_infinite}
		\max_{\sigma \in S} |c_\sigma - a_\sigma| < \gamma \implies \min_{\sigma \in S} E^c_\sigma = +\infty
	\end{equation}
	then there exist $\epsilon, \eta > 0$ such that
	\[
		\sum_{\sigma \in S} \int_0^\epsilon x^{-1 + \eta} p_\sigma(x) dx = +\infty.
	\]
	In particular $\varlimsup_{x \to 0+} \sum_{\sigma \in S} p_\sigma(x) = +\infty$.
	\item If $S$ is backward-complete and there exists $\gamma > 0$ such that
	\begin{equation}\label{eq:Ec_finite}
		\max_{\sigma \in S} |c_\sigma - a_\sigma| < \gamma \implies \max_{\sigma \in S} E^c_\sigma < +\infty	
	\end{equation}
	then there exists $\epsilon > 0$ such that
	\begin{equation}\label{eq:integrability_before_continuity}
	\sum_{\sigma \in S} \int_0^\epsilon \frac{1}{x (\log x)^2} p_\sigma(x) dx < +\infty.
	\end{equation}
	\item If Assumption~\ref{ass:continuity} is satisfied and~\eqref{eq:integrability_before_continuity} holds then $p_\sigma$ is continuous at $x =  0$ for $\sigma \in S$.
\end{enumerate}
\end{Thm}

As expected, Theorem~\ref{thm:continuity_critical_point} contains the single-field threshold as a special case.

\begin{Cor}\label{cor:single_field_threshold}
Let $S = \{\sigma_0\}$ and assume that Assumptions~\ref{ass:constant_jump_rates_no_resetting}, \ref{ass:structure_assumption_multiple_vanishing_vector_fields}, \ref{ass:integrability} and \ref{ass:continuity} are satisfied. If $\lambda_{\sigma\sigma_0} = 0$ for all $\sigma \in \Sigma_0\setminus\{\sigma_0\}$ then 
\begin{itemize}
	\item $p_{\sigma_0}$ is continuous at $x = 0$ if $\lambda_{\sigma_0} > -v_{\sigma_0}'(0)$,
	\item $p_{\sigma_0}$ is discontinuous at $x = 0$ if $\lambda_{\sigma_0} < -v_{\sigma_0}'(0)$.
\end{itemize}
\end{Cor}

\begin{proof}
	Before leaving $S = \{ \sigma_0 \}$ the process cannot jump, so $\tau$ is exponentially distributed with rate $\lambda_{\sigma_0}$ and a direct computation gives $E^c_{\sigma_0} = (\lambda_{\sigma_0} - c_{\sigma_0})^{-1}$ if $c_{\sigma_0} < \lambda_{\sigma_0}$ and $E^c_{\sigma_0} = +\infty$ otherwise. The fact that $\lambda_{\sigma\sigma_0} = 0$ for all $\sigma \in \Sigma_0 \setminus \{\sigma_0\}$ is equivalent to $S= \{\sigma_0\}$ being backward-complete.
\end{proof}

In particular, when $\sigma_t$ cannot switch between two states of $\Sigma_0$ without passing through a state in $\Sigma \setminus \Sigma_0$, Theorem~\ref{thm:continuity_critical_point} recovers the continuity threshold of~\cite[Th.~3]{bakhtin15} while dispensing with Assumption~\ref{ass:structure_assumption_single_vanishing_vector_fields}. 

\begin{Rem}
	The natural next step after studying the invariant measure is to examine the speed of convergence toward it. While this falls outside the scope of this article, we note that this question has successfully been addressed for specific PDMP models using spectral analysis~\cite{malakar18,mallmin19,das20,mischler17}, Harris-type theorems~\cite{fontbona16,evans23}, coupling~\cite{fontbona12,guillin24,hahn24} and hypocoercivity techniques~\cite{calvez14,eberle25}.
\end{Rem}

\subsection{Application to the shape transition of run-and-tumble particles}\label{subsec:application_shape_transition}

In the context of run-and-tumble particles, the cross-over from a continuous invariant
measure at high tumble rates to a discontinuous one at low tumble rates is called shape
transition. A first example is given by a single one-dimensional RTP
with two velocities in a confining potential, whose invariant measure is supported on a
compact interval and whose density diverges at the edges of this support when tumble rates
are low~\cite{dhar19}. Considering the first coordinate of a higher-dimensional RTP~\cite{basu20}, non-instantaneous tumbles~\cite{sun25} or the separation of two interacting RTPs~\cite{ledoussal21} leads to an effective particle with three velocities, whose
invariant measure can display singularities not only at the edges of its support but also in
its interior.

By comparison with the continuous Boltzmann distribution characteristic of the equilibrium
setting, continuity of the invariant measure suggests that the system is close to
equilibrium, while divergence indicates a strong deviation. The resulting dichotomy is reminiscent of the distinction between a
close-to-equilibrium and a far-from-equilibrium universality class in~\cite{hahn23} and of
the qualitative changes displayed by the invariant measure in~\cite{hahn24}. Crucially, in the statistical mechanics literature,
shape transition is classically studied by explicitly computing the invariant
measure, which seems generically intractable for systems with more than two velocities.

Our main application is the detailed picture of the shape transition undergone by the power-law process (Figure~\ref{fig:phase_diagram_powerlaw_counts}) and the harmonic process (Figure~\ref{fig:phase_diagram_harmonic_counts}). Such processes were identified in~\cite{basu20} as natural candidates for further investigation, but the explicit computation of their invariant measure seems out of reach. This emphasizes the value of the local perspective, which does not require detailed knowledge of the invariant measure. The main difficulty in the harmonic case is the joint vanishing of $v_{0_\pm}$ and $v_{0_0}$ at $x = 0$, which falls outside the scope of~\cite{bakhtin15,balazs11} and is resolved by Theorem~\ref{thm:continuity_critical_point}.

\begin{Thm}[Shape transition of the power-law process]\label{thm:instantaneous_powerlaw_process}
	The unique invariant measure of the two-particle instantaneous power-law process has the form
	\[
	\pi = \sum_{\sigma \in\Sigma} p_\sigma(x) dx \otimes \delta_\sigma
	\]
	where $p_\sigma \in L^1(\mathbb R)$ for $\sigma \in \Sigma$. All $p_\sigma$ vanish outside $[x_-, x_+]$ where $x_\pm = \pm (v/a)^{\frac{1}{p}}$. Furthermore
	\begin{itemize}
		\item $p_{\pm2} \in C^0(\mathbb R \setminus \{ x_\pm \})$ and $p_0 \in C^0(\mathbb R)$,
		\item $p_{\pm 2}$ is continuous at $x_\pm$ if and only if $\omega > a px_+^{p-1}$.
	\end{itemize}
\end{Thm}

\begin{Thm}[Shape transition of the harmonic process]\label{thm:finite_harmonic_process}
	The unique invariant measure of the two-particle finite harmonic process has the form
	\[
	\pi = \sum_{\sigma \in\Sigma} p_\sigma(x) dx \otimes \delta_\sigma
	\]
	where $p_\sigma \in L^1(\mathbb R)$. Setting $x_{\pm k} = \pm\frac{kv}{2a}$ for $k = 1, 2$, all $p_\sigma$ vanish outside $[x_{-2}, x_{+2}]$ and
	\begin{itemize}
		\item $p_{\pm k} \in C^0(\mathbb R \setminus \{ x_{\pm k}\})$ for $k = 1, 2$ and $p_{0_\pm}, p_{0_0} \in C^0(\mathbb R \setminus \{ 0 \})$,
		\item $p_{\pm2}$ is continuous at $x_{\pm2}$ if and only if $\alpha > a$,
		\item $p_{\pm 1}$ is continuous at $x_{\pm1}$ if and only if $\alpha + \beta > 2a$,
		\item $p_{0_\pm}$ is continuous (resp.~diverges) at $0$ if $\alpha > a$ (resp.~$\alpha < a$),
		\item $p_{0_0}$ is continuous (resp.~diverges) at $0$ if $\beta > a$ (resp.~$\beta < a$).
	\end{itemize}
\end{Thm}

\begin{figure}[!tb]
	\centering

	\begin{subfigure}[t]{0.48\linewidth}
		\centering
		\begin{tikzpicture}[xscale=1.06, yscale=0.57, every node/.style={font=\small}]
		
		\colorlet{cnt0}{blue!0!white}
		\colorlet{cnt2}{blue!30!white}
		
		\def\xmax{4.5}
		\def\ymax{8.6}
		
		\fill[cnt2]
		(0,0) --
		plot coordinates {
			(0.0,0.0) (0.1,0.646) (0.2,1.026) (0.3,1.344) (0.4,1.629) (0.5,1.89)
			(0.6,2.134) (0.7,2.365) (0.8,2.585) (0.9,2.797) (1.0,3.0) (1.1,3.197)
			(1.2,3.388) (1.3,3.573) (1.4,3.754) (1.5,3.931) (1.6,4.104) (1.7,4.273)
			(1.8,4.439) (1.9,4.602) (2.0,4.762) (2.1,4.92) (2.2,5.075) (2.3,5.227)
			(2.4,5.378) (2.5,5.526) (2.6,5.672) (2.7,5.817) (2.8,5.96) (2.9,6.101)
			(3.0,6.24) (3.1,6.378) (3.2,6.515) (3.3,6.65) (3.4,6.783) (3.5,6.916)
			(3.6,7.047) (3.7,7.177) (3.8,7.305) (3.9,7.433) (4.0,7.56) (4.1,7.685)
			(4.2,7.809) (4.3,7.933) (4.4,8.055) (4.5,8.177)
		} --
		(\xmax,0) -- cycle;
		
		\draw[thick]
		plot coordinates {
			(0.0,0.0) (0.1,0.646) (0.2,1.026) (0.3,1.344) (0.4,1.629) (0.5,1.89)
			(0.6,2.134) (0.7,2.365) (0.8,2.585) (0.9,2.797) (1.0,3.0) (1.1,3.197)
			(1.2,3.388) (1.3,3.573) (1.4,3.754) (1.5,3.931) (1.6,4.104) (1.7,4.273)
			(1.8,4.439) (1.9,4.602) (2.0,4.762) (2.1,4.92) (2.2,5.075) (2.3,5.227)
			(2.4,5.378) (2.5,5.526) (2.6,5.672) (2.7,5.817) (2.8,5.96) (2.9,6.101)
			(3.0,6.24) (3.1,6.378) (3.2,6.515) (3.3,6.65) (3.4,6.783) (3.5,6.916)
			(3.6,7.047) (3.7,7.177) (3.8,7.305) (3.9,7.433) (4.0,7.56) (4.1,7.685)
			(4.2,7.809) (4.3,7.933) (4.4,8.055) (4.5,8.177)
		};
		
		\draw[-{Stealth[length=2.6mm]}, thick] (-0.2,0) -- (\xmax+0.3,0) node[right] {$v/a$};
		\draw[-{Stealth[length=2.6mm]}, thick] (0,-0.2) -- (0,\ymax+0.3) node[above] {$\omega/a$};
		
		\foreach \t in {1,2,3,4} {
			\draw (\t,2pt) -- (\t,-2pt) node[below=1pt] {\t};
		}
		\foreach \t in {2,4,6,8} {
			\draw (2pt,\t) -- (-2pt,\t) node[left=1pt] {\t};
		}
		
		\node[font=\large, fill=white, circle] at (2.5,7.6) {$0$};
		\node[font=\large, fill=white, circle] at (2.5,2.5) {$2$};
		
		\end{tikzpicture}
		\caption{Power-law process for $p=3$}
		\label{fig:phase_diagram_powerlaw_counts}
	\end{subfigure}
	\hspace{0.02\linewidth}
		\begin{subfigure}[t]{0.48\linewidth}
		\centering
		\begin{tikzpicture}[scale=1.7, every node/.style={font=\small}]
		
		\colorlet{cnt0}{red!0!white}
		\colorlet{cnt1}{red!15!white}
		\colorlet{cnt2}{red!15!white}
		\colorlet{cnt3}{red!45!white}
		\colorlet{cnt5}{red!75!white}
		\colorlet{cnt6}{red!90!black}
		
		\fill[cnt0] (1,1) rectangle (2.85,2.85);                      
		\fill[cnt6] (0,0) rectangle (1,1);                            
		\fill[cnt1] (1,1) -- (2,0) -- (2.85,0) -- (2.85,1) -- cycle;   
		\fill[cnt3] (1,0) -- (2,0) -- (1,1) -- cycle;                  
		\fill[cnt3] (1,1) -- (0,2) -- (0,2.85) -- (1,2.85) -- cycle;   
		\fill[cnt5] (0,1) -- (1,1) -- (0,2) -- cycle;                  
		
		\draw[thick] (1,0) -- (1,2.85);
		\draw[thick] (0,1) -- (2.85,1);
		\draw[thick] (1,1) -- (2,0);
		\draw[thick] (1,1) -- (0,2);
		
		\draw[-{Stealth[length=2.6mm]}, thick] (-0.15,0) -- (2.95,0) node[right] {$\alpha/a$};
		\draw[-{Stealth[length=2.6mm]}, thick] (0,-0.15) -- (0,2.95) node[above] {$\beta/a$};
		
		\foreach \t in {1,2} {
			\draw (\t,2pt) -- (\t,-2pt) node[below=1pt] {\t};
			\draw (2pt,\t) -- (-2pt,\t) node[left=1pt] {\t};
		}
		
		\node[font=\large, fill=white,circle] at (2.3,2.3) {$0$};
		\node[font=\large, fill=white,circle] at (.3,.3) {$6$};
		\node[font=\large, fill=white,circle] at (2.3,0.3) {$1$};
		\node[font=\large, fill=white,circle] at (1.3,0.3) {$3$};
		\node[font=\large, fill=white,circle] at (0.3,2.3) {$3$};
		\node[font=\large, fill=white,circle] at (0.3,1.3) {$5$};
		
		\end{tikzpicture}
		\caption{Harmonic process}
		\label{fig:phase_diagram_harmonic_counts}
	\end{subfigure}
	\caption{Shape transition diagram showing the number of discontinuous $p_\sigma$}
	\label{fig:phase_diagrams_combined}
\end{figure}

\begin{Rem}
The invariant measure $\pi = \sum_{\sigma \in \Sigma} p_\sigma(x) dx \otimes \delta_\sigma$ of the single-particle resetting process can be computed explicitly using the Kolmogorov forward equation (see Lemma~\ref{lem:fokker_planck}) and is given by
\[
p_\sigma(x) = \frac{\sqrt{r(r + 2\omega)} + \sigma(\sgn x) r}{4 v}\exp\left({-\frac{\sqrt{r(r + 2\omega)}}{v} \left|x\right|}\right).
\]
Thus $p_{\pm1}$ is continuous on $\mathbb R \setminus \{ 0 \}$ and discontinuous at $0$, independently of model parameters.
\end{Rem}

\section{Regularity on noncritical intervals}\label{sec:regularity_noncritical_intervals}

To establish the regularity of invariant measures on intervals where no $v_\sigma$ vanishes, we first reformulate the generator characterization of invariance
\[
	\pi \text{ is invariant } \iff \int \mathcal L f d\pi = 0 \text{ for all } f \in D(\mathcal L),
\]
where $\mathcal L$ is the generator and $D(\mathcal L)$ its domain, as a system of linear differential equations in the sense of distributions. This is the content of Lemma~\ref{lem:fokker_planck}. We then show that all solutions of this system are regular.

\begin{Lem}[Kolmogorov forward equation]\label{lem:fokker_planck} If $\pi = \sum_{\sigma \in \Sigma} \pi_\sigma \otimes \delta_\sigma$ is invariant then for all $\sigma \in \Sigma$
	\[
	-\pi_\sigma(v_\sigma f') = \sum_{\tilde \sigma \in \Sigma} \pi_{\tilde\sigma} (\lambda_{\tilde \sigma \sigma} f) -\pi_\sigma(\lambda\R_\sigma f) + \kappa^\pi_\sigma(f) \text{ for all } f \in C^1_c(\mathbb R),
	\]
	where $\kappa^\pi = \sum_{\sigma \in \Sigma} \kappa^\pi_\sigma \otimes \delta_\sigma$ is as in Definition~\ref{def:kappa_pi}.
\end{Lem}

\begin{Rem}
	
	Writing $\pi = \sum_{\sigma \in \Sigma} \pi_\sigma \otimes \delta_\sigma$ and $\kappa^\pi = \sum_{\sigma \in \Sigma} \kappa^\pi_\sigma \otimes \delta_\sigma$ is not an assumption, as any measure $\mu$ on $E$ can be written as $\mu = \sum_{\sigma\in\Sigma}\mu_\sigma \otimes\delta_\sigma$ where the $\mu_\sigma$ are measures on $\mathbb R$. In particular, we do not assume here that $\pi$ or $\kappa^\pi$ have a density.
\end{Rem}

\begin{proof}
	For \( \sigma \in \Sigma \) let \( f_\sigma \in C^1_c(\mathbb R) \) be arbitrary but fixed. Define \( f : \mathbb R \times \Sigma \to \mathbb R \) by $f(x, \sigma) = f_\sigma(x)$. It follows from Proposition~\ref{prop:extended_generator} that \( M(t) = f(X_t) - f(X_0) - \int_0^t \mathcal L f(X_s) ds \) is a local martingale under any initial distribution where
	\[
	\mathcal L f(x, \sigma) = v_\sigma(x) \partial_x f(x, \sigma) + \sum_{\tilde\sigma \in \Sigma} \lambda_{\sigma\tilde \sigma}(x) f(x, \tilde \sigma) + \lambda\R_\sigma(x) \left( Q\R_{(x, \sigma)}(f) - f(x, \sigma)\right).
	\]	
	Note that
	\[
	\Vert \mathcal L f\Vert_\infty\le \left(\sup_{\substack{x \in K \\ \sigma \in \Sigma}} |v_\sigma(x) f_\sigma'(x)| \right) +
	\left(\max_{\sigma \in \Sigma}\sum_{\tilde\sigma \in \Sigma} \Vert \lambda_{\sigma\tilde\sigma}\Vert_{\infty}\right) \Vert f \Vert_{\infty}
	+ 2 \left(\max_{\sigma \in \Sigma}\Vert \lambda\R_\sigma \Vert_\infty  \right)\Vert f \Vert_\infty
	\]
	where \( K := \bigcup_{\sigma \in \Sigma} \mathrm{supp}(f_\sigma) \), hence
	\[
	\mathbb E_\mu \left[ \sup_{s\le t} |M(s)| \right] \le 2 ||f||_\infty + t || \mathcal L f ||_{\infty} < +\infty.
	\]
	Hence \( M(t) \) is a martingale under any initial distribution. In particular, because \( \pi \) is invariant, we have
	\[
	0 = \mathbb E_\pi \left[ f(X_t) - f(X_0) - \int_0^t \mathcal L f(X_s) ds \right] = -\mathbb E_\pi \left[\int_0^t \mathcal L f (X_s) ds \right]
	\]
	and
	\[
	\mathbb E_\pi\left[\int_0^t\mathcal L f(X_s) ds\right] = \int_0^t \mathbb E_\pi \left[ \mathcal Lf(X_s)\right] ds = t 
	\int \mathcal L f d\pi.
	\]
	Thus \( \int \mathcal L f d\pi = 0 \). Expressing this in terms of the \( f_\sigma \) we get
	\begin{align*}
	0
	&= \sum_\sigma \int \left( v_\sigma(x) f_\sigma'(x) + \sum_{\tilde\sigma}\lambda_{\sigma\tilde\sigma}(x) f_{\tilde\sigma}(x) + \lambda\R_\sigma(x) \left[ Q\R_{(x, \sigma)}(f) - f_\sigma(x) \right]\right) d\pi_\sigma(x) \\
	&= \sum_\sigma \pi_\sigma(v_\sigma f_\sigma') + \sum_\sigma \left(\sum_{\tilde\sigma} \pi_{\tilde\sigma}(\lambda_{\tilde\sigma\sigma} f_\sigma) \right)- \sum_\sigma \pi_\sigma (\lambda\R_\sigma f_\sigma) + \sum_\sigma \kappa^\pi_\sigma(f_\sigma).
	\end{align*}
	The claim now follows from the fact that the \( f_\sigma \) were arbitrary.
\end{proof}

Theorem~\ref{thm:continuity_noncritical_interval} immediately follows.

\begin{proof}[Proof of Theorem~\ref{thm:continuity_noncritical_interval}]
	Lemma~\ref{lem:fokker_planck} implies that the distribution $\varphi_{\sigma_0} := v_{\sigma_0} \pi_{\sigma_0}$ has derivative
	\[
		\sum_{\tilde \sigma \in \Sigma} \lambda_{\tilde \sigma \sigma_0} \pi_{\tilde\sigma} -\lambda\R_{\sigma_0} \pi_{\sigma_0}+ \kappa^\pi_{\sigma_0}.
	\]
	By assumption this derivative is in $L^1(I)$ so $\varphi_{\sigma_0}$ and $\pi_{\sigma_0} = \frac{1}{v_{\sigma_0}}\varphi_{\sigma_0}$ are continuous.
\end{proof}

The following lemma shows that all distributional solutions of systems of linear differential equations with regular coefficients are regular strong solutions.

\begin{Lem}\label{lem:regularity_linear_ODE} Let \( I \subset \mathbb R \) be an open interval and $k \in \mathbb N$. Further let \( A_{\sigma\tilde\sigma}, b_\sigma \in C^k(I) \) for \( \sigma, \tilde \sigma \in \Sigma \). If the family of bounded measures $(\mu_\sigma)_{\sigma \in \Sigma}$ satisfies 
	\begin{equation}\label{eq:linear_ODE}
	-\mu_\sigma(f') = \sum_{\tilde\sigma \in \Sigma} \mu_{\tilde\sigma}(A_{\sigma\tilde\sigma} f) + \int_I b_\sigma(x) f(x) dx \text{ for all } f \in C^1_c(I)
	\end{equation}
	for all $\sigma \in \Sigma$, then $\mu_\sigma$ admits a $C^{k+1}$ density on $I$ for all $\sigma \in \Sigma$.
\end{Lem}

\begin{proof} Let $x_0 \in I$ be fixed. Set $A(x) = (A_{\sigma\tilde\sigma}(x))_{\sigma,\tilde\sigma\in\Sigma}$ and let $T(x) = (T_{\sigma\tilde\sigma}(x))_{\sigma,\tilde\sigma\in\Sigma}$ be the unique $C^{k+1}$ solution of the matrix-valued differential equation
	\[
	T' = -TA
	\]
	with initial condition $T(x_0) = \text{Id}$. It follows from Grönwall's inequality that $T$ can be defined on the entire interval $I$ and $T(x_0) = \text{Id}$ implies that $T(x)$ is invertible for all $x \in I$ (see~\cite[Sec.~3.8]{taylor21}).

	Now differentiate $\sum_{\tilde\sigma} T_{\sigma\tilde\sigma} \mu_{\tilde\sigma}$ in the sense of distributions by taking $f \in C^\infty_c(I)$ and computing
	\[
	\left( \sum_{\tilde\sigma} T_{\sigma\tilde\sigma} \mu_{\tilde\sigma} \right)'(f) = -\sum_{\tilde\sigma} \mu_{\tilde\sigma}(T_{\sigma\tilde\sigma} f') = -\sum_{\tilde\sigma} \mu_{\tilde\sigma} ((T_{\sigma\tilde\sigma} f)') + \sum_{\tilde\sigma} \mu_{\tilde\sigma}(T_{\sigma\tilde\sigma}' f).
	\]
	
	Using \eqref{eq:linear_ODE} and $T_{\sigma\tilde\sigma}' = -\sum_{\hat\sigma} T_{\sigma\hat\sigma} A_{\hat\sigma\tilde\sigma}$ yields
	\begin{align*}
	\notag&\left( \sum_{\tilde\sigma} T_{\sigma\tilde\sigma} \mu_{\tilde\sigma} \right)'(f) \\
	&= \sum_{\tilde\sigma} \sum_{\hat\sigma} \mu_{\hat\sigma}(A_{\tilde\sigma\hat\sigma}T_{\sigma\tilde\sigma}f) + \sum_{\tilde\sigma} \int_I b_{\tilde\sigma}(x) T_{\sigma \tilde\sigma}(x) f(x) dx + \sum_{\tilde \sigma} \mu_{\tilde\sigma}( -\sum_{\hat\sigma} T_{\sigma\hat\sigma} A_{\hat\sigma\tilde\sigma} f) \\
	\notag&= \int_I \left( \sum_{\tilde \sigma} T_{\sigma\tilde\sigma}(x) b_{\tilde\sigma}(x)\right) f(x) dx.
	\end{align*}
	The antiderivative of a distribution is unique up to a constant. Hence there exists $C_\sigma \in \mathbb R$ s.t.
	\[
	\sum_{\tilde \sigma} T_{\sigma\tilde\sigma} \mu_{\tilde\sigma} = \left( \sum_{\tilde \sigma} \int_{x_0}^x T_{\sigma\tilde\sigma}(y) b_{\tilde\sigma}(y) dy + C_\sigma \right) dx
	\]
	has a $C^{k+1}$ density with respect to the Lebesgue measure. Because $T$ is $C^{k+1}$ and invertible, it follows that the $\mu_\sigma$ admit a $C^{k+1}$ density on $I$.
\end{proof}
	
Theorem~\ref{thm:Cr_regularity_noncritical_interval} is a direct consequence of the previous lemma.

\begin{proof}[Proof of Theorem~\ref{thm:Cr_regularity_noncritical_interval}] Set $\varphi_\sigma = v_\sigma \pi_\sigma$. It follows from Lemma~\ref{lem:fokker_planck} that for $\sigma \in \Sigma$
	\[
	-\varphi_\sigma(f') = \sum_{\tilde \sigma \in \Sigma} \varphi_{\tilde\sigma} \left(\frac{\lambda_{\tilde \sigma \sigma}}{v_{\tilde\sigma}} f\right) -\varphi_\sigma\left(\frac{\lambda\R_\sigma }{v_\sigma}f\right) + \int_I f(x) k^\pi_\sigma(x) dx \text{ for all } f \in C^1_c(I)
	\]
	where $k^\pi_\sigma \in C^{r-1}(I)$ is the density of $\kappa^\pi_\sigma$. 
	Lemma~\ref{lem:regularity_linear_ODE} then implies that each $\varphi_\sigma$, and consequently each $\pi_\sigma = \frac{1}{v_\sigma}\,\varphi_\sigma$, possesses a $C^r$ density on~$I$.
\end{proof}

\section{Continuity at multiply critical points}\label{sec:continuity_critical_point}

In this section we prove Theorem~\ref{thm:continuity_critical_point}, working under Assumptions~\ref{ass:constant_jump_rates_no_resetting} and~\ref{ass:structure_assumption_multiple_vanishing_vector_fields} throughout. As noted in Section~\ref{subsec:intro_C0_critical_points}, the continuity of the invariant densities is determined by the behavior of the integrals
\[
I_\mathrm{d}(\epsilon, \eta) = \sum_{\sigma \in S} \int_0^\epsilon x^{-1+\eta} p_\sigma(x)dx,
\]
and  
\[
I_\mathrm{c}(\epsilon) = \sum_{\sigma \in S} \int_0^\epsilon \frac{1}{x (\log x)^2} p_\sigma(x)dx.
\]
Divergence of $I_\mathrm{d}(\epsilon, \eta)$ implies discontinuity, while the finiteness of $I_\mathrm{c}(\epsilon)$ guarantees continuity under suitable asymptotics for $p_\sigma$. These integrals are analyzed by linearizing the deterministic dynamics around $x =  0$ and linking them to the expectations
\[
	E^{c}_\sigma = \mathbb E_\sigma \left[ \int_0^\tau e^{\int_0^t c_{\sigma_s}ds} dt \right].
\]

\subsection{Proof of Theorem~\ref{thm:continuity_critical_point} (i)}\label{sec:c0_critical_point_i}

The following classical representation of the invariant measure~\eqref{eq:stopping_time_representation} is at the heart of the link between $I_\mathrm{d}(\epsilon, \eta), I_\mathrm{c}(\epsilon)$ and $E^c_\sigma$.

\begin{Def}[Induced Markov chain] Let $\epsilon > 0$. Set $\tau_0 = 0$ as well as
	\begin{align*}
	\tilde \tau_n &= \inf \{ t \ge \tau_{n-1} : X_t \notin [0,\epsilon] \times S\}, \\
	\tau_n  &= \inf \{t > \tilde \tau_n : X_t\in [0, \epsilon]\times S\},
	\end{align*}
	for all $n \ge 1$ and define $Z_n = X_{\tau_n}$.
\end{Def}

It follows from the strong Markov property of $X_t$ that $Z_n$ is a Markov chain (provided it is well defined, i.e.~$\tau_n < +\infty$ a.s.~for all $n \ge 0$). Its state space is $E_Z := [0, \epsilon] \times S$.

\begin{Lem}\label{lem:stopping_time_representation}
	Under Assumption~\ref{ass:integrability}, there exists $\delta > 0$ such that for all $\epsilon \in (0, \delta)$:
	\begin{enumerate}[label=(\roman*)]
		\item One has that $\mathbb E_\mu[\tau_1] < +\infty$ for all measures $\mu$ on $E_Z$ (in particular $Z_n$ is well defined).
		\item The Markov chain $Z_n$ admits a unique invariant measure $\pi_Z$.
		\item For all bounded (or positive) measurable $f : E \to \mathbb R$ one has
		\begin{align}\label{eq:stopping_time_representation}
		\pi(f) = \frac{1}{\mathbb E_{\pi_Z}\left[ \tau_1 \right]} \mathbb E_{\pi_Z} \left[\int_0^{\tau_1} f(X_t) dt\right]
		\end{align}
		where $\pi$ is the unique invariant measure of $X_t$.
	\end{enumerate}
\end{Lem}

\begin{proof} (i) One has $\mathbb E_{\mu}[\tilde\tau_1] \le \mathbb E_{\mu}[\tau]$ where $\tau = \inf \{ t \ge 0 : \sigma_t \notin S\}$. Because $\sigma_t$ is an irreducible Markov jump process with finite state space and $S \ne \Sigma$, we have $\mathbb E_{\mu}[\tau] < +\infty$.
	
As $0 \in \mathring K$, there exists $\delta > 0$ such that $[0, \epsilon] \subset K$ for all $\epsilon \in (0, \delta)$. By the strong Markov property we have
	\[
	\mathbb E_{\mu} [\tau_1]= \mathbb E_{\mu}[\tilde\tau_1] + \mathbb E_{\mu}[\mathbb E_{X_{\tilde\tau_1}}[\bar\tau]] \le \mathbb E_{\mu}[\tilde\tau_1] + \sup_{(x, \sigma) \in [0, \epsilon] \times \Sigma} \mathbb E_{(x, \sigma)}[\bar\tau]
	\]
	where $\bar \tau = \inf \{t > 0 : X_t \in [0, \epsilon] \times S\}$ using that $X_{\bar\tau} \in [0, \epsilon] \times \Sigma$ a.s.~(by the continuity of $x_t$).

	Because $\pi([0, \epsilon] \times S) > 0$ and $\sup_{(x, \sigma) \in K \times \Sigma} \Vert \delta_{(x, \sigma)} P_t - \pi \Vert_\mathrm{TV} \to 0$ when $t \to +\infty$, there exists $T > 0$ such that
	\[
	\inf_{(x,  \sigma) \in K \times\Sigma} \mathbb P_{(x, \sigma)} \left( X_T \in [0, \epsilon] \times S \right) =: p > 0.
	\]
	
	Because $
	\phi^\sigma_t(K) \subset K \text{ for all } t \ge 0 \text{ and } \sigma \in \Sigma
	$, one has $\mathbb P_{(x, \sigma)}(X_T \in K \times\Sigma) = 1$ for all $(x, \sigma) \in K\times \Sigma$. It follows from the Markov property that $\mathbb P_{(x, \sigma)}(\bar\tau > kT) \le (1 - p)^k$ for all $(x, \sigma) \in K \times \Sigma$. Thus
	\[
	\mathbb E_{(x, \sigma)}[\bar\tau] = T \mathbb E_{(x, \sigma)}[\bar\tau/T] = T \int_0^{+\infty} \mathbb P_{(x, \sigma)}(\bar\tau/T > s) ds \le T \sum_{k = 0}^{+\infty} \mathbb P_{(x, \sigma)}(\bar\tau/T > k) \le \frac{T}{p}
	\]
	for all $(x, \sigma) \in K \times \Sigma$. Putting everything together, we get $\mathbb E_{\mu} [\tau_1] < +\infty$.

	(ii) It suffices to show that $Z_n$ satisfies the Doeblin condition. Let $\sigma^* \in \Sigma \setminus S$ be such that there exists $\sigma \in S$ with $\lambda_{\sigma\sigma^*} > 0$ and assume without loss of generality that $v_{\sigma^*}(0) > 0$. Under Assumption~\ref{ass:integrability}, there exists $\delta > 0$ such that for all $\epsilon \in (0, \delta)$ one has $v_\sigma(x) \le 0$ for $(x, \sigma) \in E_Z$. Assume without loss of generality that $v_{\sigma^*}(x) > 0$ for $x \in [0, \epsilon]$. Setting $\hat\tau = \inf \{ t \ge 0 : X_t = (\epsilon, \sigma^*)\}$, one has that for all positive measurable $f : E_Z \to \mathbb R$
	\begin{align*}
		\mathbb E_{(x, \sigma)}[f(X_{\tau_1})] \ge \mathbb E_{(x, \sigma)} \left[ 1_{\{\hat\tau< \tau_1\}} \mathbb E_{X_{\hat\tau}} [f(X_{\tau_1})]\right] = \mathbb P_{(x, \sigma)} \left( \hat\tau< \tau_1 \right) \mathbb E_{(\epsilon, \sigma^*)} \left[ f(X_{\tau_1}) \right]
	\end{align*} 
	by the strong Markov property. Thus it suffices to show $\inf_{(x, \sigma) \in E_Z} \mathbb P_{(x, \sigma)} \left( \hat\tau< \tau_1 \right) > 0$ to show the Doeblin property.
	
	Starting from an initial position $(x, \sigma) \in [0, \epsilon] \times \{\sigma^*\}$, if $\sigma_t$ does not jump before time $\epsilon/\left( \inf_{0 \le x \le \epsilon} v_{\sigma^*}(x) \right)$ then $X_t$ passes through the state $(\epsilon, \sigma^*)$ before the stopping time $\tau_1$. Hence
	\[
	\mathbb P_{(x, \sigma)} \left(\hat\tau< \tau_1\right) \ge \mathbb P_{(x, \sigma)} \left( X_{\tilde\tau_1} \in [0, \epsilon] \times \{ \sigma^* \} \right) e^{-\lambda_{\sigma^*} \epsilon/\left( \inf_{0 \le x \le \epsilon} v_{\sigma^*}(x) \right)}.
	\]
	The assertion now follows from the fact that $\inf_{(x, \sigma) \in E_Z} \mathbb P_{(x, \sigma)}(X_{\tilde\tau_1} \in [0, \epsilon] \times \{ \sigma^* \}) > 0$ is implied by the irreducibility of $S$.
	
	(iii)  It follows from~\cite[Th.~6.26]{benaim22course} that the right hand side of \eqref{eq:stopping_time_representation} is an invariant measure. Hence~\eqref{eq:stopping_time_representation} follows from the uniqueness of $\pi$.
\end{proof}

Theorem~\ref{thm:continuity_critical_point} (i) is derived by linearizing the deterministic dynamics around $0$.

\begin{proof}[Proof of Theorem~\ref{thm:continuity_critical_point}~(i)]
	
	Fix $\delta \in (0, \min_{\sigma \in S} a_\sigma)$ to be chosen later. Because $v_\sigma(x) = -a_\sigma x + o(x)$ when $x \to 0+$, there exists $\epsilon> 0$ such that for all $\sigma \in S$ one has
	\[
			v_\sigma(x) \le (-a_\sigma + \delta) x \le 0 \text{ for all } x \in [0, \epsilon].
	\]
	Grönwall's inequality implies
	\[
		x_t \le x_0 e^{\int_0^t (-a_{\sigma_s} + \delta) ds }.
	\]
	
	Fix $\eta \in (0, 1)$ to be chosen later. Taking $\epsilon$ smaller if necessary, one has by Lemma~\ref{lem:stopping_time_representation}
	\begin{align*}
	\sum_{\sigma\in S} \int_0^\epsilon x^{-1 + \eta} p_\sigma(x) dx &= \frac{1}{\mathbb E_{\pi_Z}\left[ \tau_1 \right]} \mathbb E_{\pi_Z} \left[\int_0^{\tau_1} 1_{\{X_t \in [0, \epsilon] \times S\}} x_t^{-1+\eta} dt\right] \\
	&= \frac{1}{\mathbb E_{\pi_Z}\left[\tau_1 \right]} \mathbb E_{\pi_Z} \left[\int_0^{\tilde \tau_1} x_t^{-1+\eta} dt\right] \\
	&\ge \frac{1}{\mathbb E_{\pi_Z}\left[\tau_1 \right]} \mathbb E_{\pi_Z} \left[\int_0^{\tilde \tau_1} x_0^{-1 + \eta} e^{\int_0^t (-1 + \eta) (-a_{\sigma_s} + \delta) ds}dt\right] \\
	&= \frac{1}{\mathbb E_{\pi_Z}\left[\tau_1 \right]}  \mathbb E_{\pi_Z} \left[\int_0^{\tau} x_0^{-1 + \eta} e^{\int_0^t (-1 + \eta) (-a_{\sigma_s} + \delta) ds}dt\right]
	\end{align*}
	using the fact that $v_\sigma(x) \le 0$ for $(x, \sigma) \in [0, \epsilon] \times S$ implies that $\tilde \tau_1 = \tau$ for the last equality.
	
	Choosing $\delta$ and $\eta$ such that $c_\sigma := (-1 + \eta)(-a_\sigma + \delta)$ satisfies $\max_{\sigma \in S} |c_\sigma - a_\sigma| < \gamma$ yields
	\[
		\mathbb E_{\pi_Z} \left[\int_0^{\tau} x_0^{-1 + \eta} e^{\int_0^t c_{\sigma_s} ds}dt \right] \ge \underbrace{\mathbb E_{\pi_Z} \left[ x_0^{-1+\eta} \right]}_{>0} \underbrace{\left( \min_{\sigma \in S} E^c_\sigma \right)}_{=+\infty}.
	\]
\end{proof}

\subsection{Proof of Theorem~\ref{thm:continuity_critical_point} (ii)}

	Using the same ideas as in Section~\ref{sec:c0_critical_point_i}, one can show
	\[
		I_\mathrm{c}(\epsilon) \le \frac{1}{\mathbb E_{\pi_Z}\left[ \tau_1 \right]} \mathbb E_{\pi_Z} \left[ \frac{1}{x_0 (\log x_0)^2} \right]\left(\max_{\sigma \in S} E^c_\sigma\right)
	\]
	for suitably chosen $c_\sigma$. Assuming that the $E^c_\sigma$ are finite, the finiteness of $I_\mathrm{c}(\epsilon)$ follows from the next lemma, whose proof is presented after that of Theorem~\ref{thm:continuity_critical_point} (ii).
	
\begin{Lem}\label{lem:x0_integrability}
	If Assumption~\ref{ass:integrability} is satisfied and $S$ is backward-complete then
	\[
		\mathbb E_{\pi_Z} \left[ \frac{1}{x_0 (\log x_0)^2} \right] < +\infty.
	\]
\end{Lem}

\begin{proof}[Proof of Theorem~\ref{thm:continuity_critical_point} (ii)]
	Because $v_\sigma(x) = -a_\sigma x + o(x)$ when  $x \to 0+$, there exists $\epsilon > 0$ such that for all $\sigma \in S$
	\[
		v_\sigma(x) \ge (-a_\sigma - \gamma/2) x \text{ and } v_\sigma(x) \le 0 \text{ for all } x \in [0, \epsilon].
	\]
	Hence it follows from the comparison principle for ODEs that for $x_0 \in [0, \epsilon]$ and $t \le \tilde \tau_1$ one has
	\[
		x_t \ge x_0 \underbrace{e^{\int_0^t -a_{\sigma_s} -\gamma/2 ds}}_{=: e_t}.
	\]
Choosing $\epsilon > 0$ smaller if necessary, one may assume that $x \mapsto 1/[x (\log x)^2]$ is decreasing on~$[0, \epsilon]$. Hence
	\[
	\frac{1}{x_t (\log x_t)^2} \le \frac1{x_0 e_t (\log (x_0 e_t))^2} \le \frac{1}{x_0 (\log x_0)^2 e_t} \text{ for all } t \le \tilde \tau_1
	\]
	using the fact that $-a_\sigma - \gamma/2 < 0$ for all $\sigma \in S$ implies $e_t \le 1$ for the second inequality. Taking $\epsilon > 0$ smaller if necessary, it follows from Lemma~\ref{lem:stopping_time_representation} that
	\begin{align*}
	\sum_{\sigma \in S} \left[\int_0^\epsilon \frac{1}{x (\log x)^2} p_\sigma(x) dx\right] &= \frac{1}{\mathbb E_{\pi_Z}\left[ \tau_1 \right]} \mathbb E_{\pi_Z} \left[\int_0^{\tau_1} \frac{1}{x_t (\log x_t)^2} 1_{\{X_t \in [0, \epsilon] \times S\}} dt\right] \\
	&= \frac{1}{\mathbb E_{\pi_Z}\left[ \tau_1 \right]} \mathbb E_{\pi_Z} \left[\int_0^{\tilde\tau_1} \frac{1}{x_t (\log x_t)^2} dt\right] \\
	&= \frac{1}{\mathbb E_{\pi_Z}\left[ \tau_1 \right]} \mathbb E_{\pi_Z} \left[\int_0^{\tau} \frac{1}{x_t (\log x_t)^2} dt\right] \\
	&\le \frac{1}{\mathbb E_{\pi_Z}\left[ \tau_1 \right]} \mathbb E_{\pi_Z} \left[\int_0^{\tau} \frac{1}{x_0 (\log x_0)^2 e_t} dt\right] \\
	&\le \frac{1}{\mathbb E_{\pi_Z}\left[ \tau_1 \right]} \mathbb E_{\pi_Z} \left[ \frac{1}{x_0 (\log x_0)^2} \right]\left(\max_{\sigma \in S} \mathbb E_\sigma \left[\int_0^{\tau} \frac{1}{e_t} dt\right]\right).
	\end{align*}
By Lemma~\ref{lem:x0_integrability} one has
\[
	\mathbb E_{\pi_Z} \left[ \frac{1}{x_0 (\log x_0)^2} \right] < +\infty.
\]
Taking $c_\sigma = a_\sigma+ \gamma/2$ one has
\[
\max_{\sigma \in S} \mathbb E_\sigma \left[\int_0^{\tau} \frac{1}{e_t}dt\right] = \max_{\sigma \in S} E^c_\sigma < +\infty
\]
by assumption.
\end{proof}

\begin{proof}[Proof of Lemma~\ref{lem:x0_integrability}]
	Denote
	\begin{align*}
	T_0 &= \tilde \tau_1, & T_1 &= \inf\{t\ge T_0: \sigma_{t-} \ne \sigma_{t+}\},  & T_2 &= \inf\{t\ge T_1: \sigma_{t-} \ne \sigma_{t+}\},
	\end{align*}
	and so on the jumps of the velocity state $\sigma_t$ after the time $\tilde \tau_1$. The invariance of $\pi_Z$ implies
	\begin{align*}
	\mathbb E_{\pi_Z} \left[ \frac{1}{x_0 (\log x_0)^2}\right] &= \mathbb E_{\pi_Z} \left[ \frac{1}{x_{\tau_1} (\log x_{\tau_1})^2}\right] \\
	&= \mathbb E_{\pi_Z} \left[ 1_{\{\tau_1 \notin \{T_0, T_1, \ldots\}\}} \frac{1}{x_{\tau_1} (\log x_{\tau_1})^2}\right] + \sum_{k = 1}^{+\infty} \mathbb E_{\pi_Z} \left[ 1_{\{\tau_1 = T_k\}} \frac{1}{x_{\tau_1} (\log x_{\tau_1})^2}\right].
	\end{align*}
	
	If $\tau_1 \notin \{ T_0, T_1, \ldots \}$ then $x_{\tau_1} = \epsilon$ and thus
	\[
	\mathbb E_{\pi_Z} \left[ 1_{\{\tau_1 \notin \{T_0, T_1, \ldots\}\}} \frac{1}{x_{\tau_1} (\log x_{\tau_1})^2}\right] = \mathbb P_{\pi_Z} \left( \tau_1 \notin \{T_0, T_1, \ldots\}\right) \frac{1}{\epsilon (\log \epsilon)^2} < +\infty.
	\]
	
	Furthermore, denoting $\bar\tau_1 = \inf \{t > 0: X_t \in [0, \epsilon] \times S\}$ and $\bar T_1 = \inf\{t \ge 0 : \sigma_{t-} \ne \sigma_{t+}\}$, it follows from the strong Markov property that
	\begin{align*}
	\sum_{k = 1}^{+\infty} \mathbb E_{\pi_Z} \left[ 1_{\{\tau_1 = T_k\}} \frac{1}{x_{\tau_1} (\log x_{\tau_1})^2}\right] &= \sum_{k = 1}^{+\infty} \mathbb E_{\pi_Z} \left[ 1_{\{\tau_1 > T_{k-1}\}} \mathbb E_{X_{T_{k-1}}} \left[ 1_{\{\bar\tau_1 = \bar T_1\}}\frac{1}{x_{\bar \tau_1} (\log x_{\bar \tau_1})^2}\right]\right] \\
	&\le \left( \sum_{k = 1}^{+\infty} \mathbb P_{\pi_Z} \left( \tau_1 > T_{k-1}\right) \right) \sup_{(x, \sigma) \notin [0, \epsilon] \times S} \mathbb E_{(x, \sigma)} \left[ \frac{1_{\{\bar\tau_1 = \bar T_1\}}}{x_{\bar \tau_1} (\log x_{\bar \tau_1})^2} \right].
	\end{align*}
	
	Fix $a > 0$ to be chosen later. One has
	\begin{align*}
	\sum_{k = 1}^{+\infty} \mathbb P_{\pi_Z} \left( \tau_1 > T_{k-1}\right) &\le \sum_{k = 0}^{+\infty} \mathbb P_{\pi_Z} \left( \tau_1 > a k \right) + \sum_{k = 0}^{+\infty} \mathbb P_{\pi_Z} \left( T_k < a k \right) \\
	&= \mathbb E_{\pi_Z} \left[ \sum_{k = 0}^{+\infty} 1_{\{k < \tau_1/a\}}  \right] + \sum_{k = 0}^{+\infty} \mathbb P_{\pi_Z} \left( T_k < a k \right) \\
	&\le \underbrace{\mathbb E_{\pi_Z} \left[  \tau_1/a + 1 \right]}_{< +\infty} + \sum_{k = 0}^{+\infty} \mathbb P_{\pi_Z} \left( T_k < a k \right)
	\end{align*}
	
	Let $E_i$ be independent exponential random variables with rate $\max_{\sigma \in \Sigma} \lambda_\sigma$. Then, by stochastic domination, one has
	\[
	\mathbb P_{\pi_Z} \left(T_k < ak\right) \le \mathbb P\left( \sum_{i = 0}^{k-1} E_i < ak \right) = \mathbb P\left( \frac1k \sum_{i = 0}^{k-1} E_i < a \right).
	\]
	Hence, if $a < \mathbb E[E_i]$,  Chernoff bounds show that $\mathbb P\left( \frac1k \sum_{i = 0}^{k-1} E_i < a \right)$ decays exponentially with $k$ and hence
	\[
	\sum_{k = 0}^{+\infty} \mathbb P_{\pi_Z} \left( T_k < a k \right) < +\infty.
	\]
	
	It remains to show 
	\[
	\sup_{(x, \sigma) \notin [0, \epsilon] \times S} \mathbb E_{(x, \sigma)} \left[  \frac{1_{\{\bar\tau_1 = \bar T_1\}}}{x_{\bar \tau_1} (\log x_{\bar \tau_1})^2} \right] < +\infty.
	\]
	
	Let $\sigma \in \Sigma$ be arbitrary but fixed. Distinguish between the following cases
	\begin{itemize}
		\item \underline{Case $(x, \sigma) \in (\mathbb R \setminus [0, \epsilon]) \times S$.} The definition of $\bar\tau_1$ implies that $x_t \notin [0, \epsilon]$ for $t < \bar\tau_1$ and $x_t \in [0, \epsilon]$ for $t > \bar\tau_1$ close enough to $\bar\tau_1$. Together with the continuity of $t \mapsto x_t$ this implies $x_{\bar\tau_1} \in \{0, \epsilon\}$. Because $v_\sigma(0) = 0$, $x_{\bar\tau_1} = 0$ would imply $x_t = 0$ for all $t < \bar\tau_1$. This is absurd. Hence $x_{\bar\tau_1} = \epsilon$. It follows
		\[
			\mathbb P_{(x, \sigma)} (\bar\tau_1= \bar T_1) \le \mathbb P_{(x, \sigma)} (\bar T_1 \in \{ t \ge 0 : \phi^\sigma_t(x) = \epsilon \}) = 0
		\]
		using the fact that $\{ t \ge 0 : \phi^\sigma_t(x) = \epsilon \}$ is either a singleton or the empty set.
		\item \underline{Case $\sigma \in \Sigma_0 \setminus S$.} Then the backward-completeness of $S$ implies that one cannot go from $\sigma \in \Sigma_0 \setminus S$ to any state in $S$ in one jump. Thus $\mathbb P_{(x, \sigma)}(\bar \tau_1 = \bar T_1) = 0$.
		
		\item \underline{Case $\sigma \notin \Sigma_0$.} Because $v_\sigma(0) \ne 0$ there exists $\tilde \epsilon > 0$ such that $\inf_{x \in [0, \tilde \epsilon]} |v_\sigma(x)| > 0$. One has
		\[
			\mathbb E_{(x, \sigma)} \left[ \frac{1_{\{\bar\tau_1 = \bar T_1\}}}{x_{\bar \tau_1} (\log x_{\bar \tau_1})^2} \right] \le \mathbb E_{(x, \sigma)} \left[ \frac{1_{\{\bar\tau_1 = \bar T_1\}} 1_{\{x_{\bar \tau_1} \le \tilde \epsilon\}}}{x_{\bar \tau_1} (\log x_{\bar \tau_1})^2} \right] + \frac{1}{\tilde\epsilon (\log \tilde\epsilon)^2}.
		\]
		
		Define $\varphi(t) = \phi_t^\sigma(x)$. One has
		\begin{align*}
		&\mathbb E_{(x, \sigma)} \left[ \frac{1_{\{\bar\tau_1 = \bar T_1 \}} 1_{\{x_{\bar \tau_1} \le \tilde \epsilon\}}}{x_{\bar \tau_1} (\log x_{\bar \tau_1})^2} \right] \\
		&= \int_0^{+\infty} \frac{1_{\{\varphi(t) \in [0, \tilde \epsilon]\}}}{\varphi(t) (\log \varphi(t))^2}  \lambda_\sigma e^{-\lambda_\sigma t} \left( \sum_{\tilde\sigma \in S} \frac{\lambda_{\sigma\tilde\sigma}}{\lambda_\sigma}\right) dt \\
		&\le \left( \sum_{\tilde\sigma \in S} {\lambda_{\sigma\tilde\sigma}}\right) \int_0^{+\infty} \frac{1_{\{\varphi(t) \in [0,\tilde \epsilon]\}}}{\varphi(t) (\log \varphi(t))^2} dt.
		\end{align*}
	\end{itemize}

	If $\varphi(t) = x$ is constant then $v_\sigma(x) = 0$ so $\inf_{x \in [0, \tilde\epsilon]} |v_\sigma(x)| > 0$ implies that $x \notin [0, \tilde\epsilon]$. Hence $1_{\{\varphi(t) \in [0,\tilde \epsilon]\}} = 0$ for all $t \ge 0$.

	Because $v_\sigma$ is Lipschitz, if $v_\sigma(x) \ne 0$ then $\varphi(t) = \phi^\sigma_t(x)$ is a $C^1$ diffeomorphism from $(0, +\infty)$ to its image. One can thus make the change of variable $y = \varphi(t)$ and get
	\begin{align*}
	\int_0^{+\infty} \frac{1_{\{\varphi(t) \in [0, \tilde\epsilon]\}}}{\varphi(t) (\log \varphi(t))^2}  dt &= \int_{\varphi(0)}^{\lim_{t\to+\infty}\varphi(t)} \frac{1_{\{y \in [0, \tilde\epsilon]\}}}{y (\log y)^2} \frac{1}{\varphi'(\varphi^{-1}(y))} dy \\
	&\le \frac{1}{\inf_{x \in [0, \tilde\epsilon]} |v_\sigma(x)|} \int_0^{\tilde\epsilon} \frac1{y (\log y)^2} dy < +\infty
	\end{align*}
	using the fact that $|\varphi'(\varphi^{-1}(y))|^{-1} \le (\inf_{x \in [0, \tilde\epsilon]} |v_\sigma(x)|)^{-1}$ for all $y \in [0, \tilde \epsilon]$.
\end{proof}

\subsection{Proof of Theorem~\ref{thm:continuity_critical_point} (iii)}

The key observation to deduce continuity from Theorem~\ref{thm:continuity_critical_point} (ii) is to notice that if $\sigma^* \in S$ is s.t.
\begin{equation}\label{eq:pSigma_asymptotic}
p_{\sigma^*}(x) \underset{x \to 0}{\sim} C x^\nu (\log x)^k \text{ with } C \ne 0, \nu \in \mathbb R \text{ and } k \in \mathbb N
\end{equation}
then
\[
\sum_{\sigma \in S} \int_0^\epsilon \frac{1}{x (\log x)^2} p_\sigma(x) dx < +\infty
\]
implies $\nu > 0$ or $\nu = k = 0$. In particular $p_{\sigma^*}$ admits a limit to the right. The following lemma establishes a slightly weakened version of~\eqref{eq:pSigma_asymptotic}.

\begin{Lem}\label{lem:asymptotics_for_analytic_vector_fields}
	Assume that the $v_\sigma$ are all analytic at $0$ and $a_\sigma := -v'_\sigma(0) \ne 0$ for all $\sigma \in \Sigma_0$. If $B_0$ (defined as in Assumption~\ref{ass:continuity}) is diagonalizable and has only real eigenvalues then for all $\sigma \in \Sigma$
	\[
		p_\sigma(x) = o(1) \text{ or } p_\sigma(x) = C x^{\nu} (\log x)^k + o(x^\nu (\log x)^k) + o(1) \text{ when } x \to 0+
	\]
	where $C \in \mathbb R^*$, $\nu \in \mathbb R$ and $k \in \mathbb N$.
\end{Lem}

\begin{proof} As in~\cite[Sec.~7.2]{bakhtin15}, it follows from~\cite[Prop.~3.11.7]{taylor21} that there exists a nilpotent matrix $\Gamma$ such that
	\[
		\varphi(x) = (\varphi_\sigma(x))_{\sigma \in \Sigma} = (v_\sigma(x) p_\sigma(x))_{\sigma\in\Sigma}
	\]
	is given by
    \[
        \varphi(x) = U(x) x^{B_0} x^\Gamma u
    \]
    where $U(x) = \sum_{n =0}^{+\infty} U_n x^n$ is a matrix-valued analytic function, $u = (u_\sigma)_{\sigma \in \Sigma}$ is a vector and $x^{B_0}$ (resp.~$x^\Gamma$) stands for the matrix $e^{(\log x) B_0}$ (resp.~$e^{(\log x) \Gamma}$). Setting $N = |\Sigma|$, one has $\Gamma^N = 0$ so the entries of $x^\Gamma u$ are of the form
    \[
        \sum_{n = 0}^{N-1} a_n (\log x)^n
    \]
    where $a_n \in \mathbb R$ and the entries of $x^{B_0} x^\Gamma u$ are of the form
    \[
        \sum_{n = 0}^{N-1} \sum_{m = 0}^{N - 1} a_{nm} x^{\kappa_n} (\log x)^m
    \]
    where $a_{nm} \in \mathbb R$ and $\kappa_0, \ldots, \kappa_{N-1}$ are the real eigenvalues of $B_0$.

    Taking $M > 1 - \min \kappa_n$ and writing
    \[
        \varphi(x) = \left(\sum_{n = 0}^{M-1} U_n x^n \right) x^{B_0} x^{\Gamma} u + \underbrace{\left(\sum_{n = M}^{+\infty} U_n x^n \right) x^{B_0} x^{\Gamma} u}_{= o(x)}
    \]
    it follows that $\varphi_\sigma(x) = o(x)$ or $\varphi_\sigma(x) = C x^{\nu} (\log x)^k + o(x^{\nu} (\log x)^k) + o(x)$ for $C \ne 0$, $\nu \in \mathbb R$ and $k \in \mathbb N$. The result now follows from $p_\sigma = \varphi_\sigma / v_\sigma$.
\end{proof}

Proving Theorem~\ref{thm:continuity_critical_point}~(iii) now reduces to verifying that the left and right limits of $p_\sigma$ coincide.

\begin{proof}[Proof of Theorem~\ref{thm:continuity_critical_point} (iii)] Let $\sigma \in S$ be arbitrary but fixed. If $p_\sigma(x) = o(1)$ when $x \to 0+$ then $p_\sigma$ admits a limit to the right at $x = 0$. If not, Lemma~\ref{lem:asymptotics_for_analytic_vector_fields} implies that $p_\sigma(x) = C x^\nu (\log x)^k + o(x^\nu (\log x)^k) + o(1)$ when $x \to 0+$. It follows from Theorem~\ref{thm:continuity_critical_point} (ii) that there exists $\epsilon > 0$ such that
    \[
        \int_0^\epsilon \frac{1}{x (\log x)^2} p_\sigma(x)  dx < +\infty.
    \]
    This implies $\nu > 0$ or $\nu = k = 0$ so $p_\sigma(x)$ admits a limit to the right.
    
    If $\pi([-\epsilon, 0] \times \{ \sigma \}) > 0$ for all $\epsilon > 0$ then it follows from the argument above that $p_\sigma$ also admits a limit to the left. If there exists $\epsilon > 0$ such that $\pi([-\epsilon, 0] \times \{ \sigma \}) = 0$ then $\lim_{x \to 0-} p_\sigma(x) = 0$. In both cases $p_\sigma$ admits a limit to the left at $x = 0$.

    It remains to show that the limits to the left and right of $x = 0$ coincide. By Theorem~\ref{thm:Cr_regularity_noncritical_interval}, the analyticity of the $v_\sigma$ at $x = 0$ implies that there exists $\epsilon > 0$ such that $p_\sigma \in C^\infty(0, \epsilon)$ for all $\sigma \in\Sigma$. Hence it follows from Lemma~\ref{lem:fokker_planck} that
    \[
        -(v_\sigma p_\sigma)' + \sum_{\tilde\sigma \in \Sigma} \lambda_{\tilde\sigma\sigma} p_{\tilde\sigma} = -(v_\sigma p_\sigma)' + \sum_{\tilde\sigma \in S \cup S_\mathrm{in}} \lambda_{\tilde\sigma\sigma} p_{\tilde\sigma} = 0 \text{ for all } \sigma \in S
    \]
    in the strong sense on $(0, \epsilon)$.

    By assumption $S_\mathrm{in}\cap \Sigma_0 = \emptyset$ hence Theorem~\ref{thm:continuity_noncritical_interval} implies that $p_\sigma$ is continuous at $x = 0$ for $\sigma \in S_\mathrm{in}$. It follows that the limit of
    \[
    (v_\sigma p_\sigma)' = \sum_{\tilde\sigma \in S \cup S_\mathrm{in}} \lambda_{\tilde\sigma\sigma} p_{\tilde\sigma}
    \]
    when $x\to0+$ exists. One has
    \begin{align*}
        \lim_{x \to 0+}(v_\sigma p_\sigma)'(x) &= \lim_{h \to 0+} \frac{v_\sigma(2h) p_\sigma(2h) - v_\sigma(h)p_\sigma(h)}{h} \\
        &= \lim_{h\to0+} \frac{v_\sigma(2h) - v_\sigma(h)}{h} p_\sigma(2h) + \lim_{h \to 0+} \frac{v_\sigma(h)}{h} \left( p_\sigma(2h) - p_\sigma(h) \right) \\
        &= v_\sigma'(0) p_\sigma(0+).
    \end{align*}

    It follows
    \[
        a_\sigma p_\sigma(0+) + \sum_{\tilde\sigma \in S \cup S_\mathrm{in}} \lambda_{\tilde\sigma\sigma} p_{\tilde\sigma}(0+) = 0
    \]
    so the invertibility of $A$ implies
    \[
    (p_\sigma(0+))_{\sigma \in S \cup S_\mathrm{in}} = A^{-1} \left( p_\sigma(0+) 1_{\{\sigma \in S_\mathrm{in}\}}\right)_{\sigma \in S\cup S_\mathrm{in}}.
    \]
    The same arguments imply
    \[
    (p_\sigma(0-))_{\sigma \in S \cup S_\mathrm{in}} = A^{-1} \left( p_\sigma(0-) 1_{\{\sigma \in S_\mathrm{in}\}}\right)_{\sigma \in S\cup S_\mathrm{in}}.
    \]
    The result now follows from the continuity of $p_\sigma$ for $\sigma \in S_\mathrm{in}$.
\end{proof}

\subsection{\texorpdfstring{Explicit computation of $E^{c}_\sigma$}{Explicit computation of Ecσ}}\label{sec:Ec_finiteness}

The goal of Theorem~\ref{thm:continuity_critical_point} is to obtain a threshold for the transition rates above which the invariant densities are continuous and below which they diverge. To achieve this, the conditions 
\[
\max_{\sigma \in S} E^c_\sigma  < +\infty \text{ and } \min_{\sigma \in S} E^c_\sigma = +\infty
\]
have to be made explicit in terms of model parameters. The upcoming Lemma~\ref{lem:Ec_ec_general_properties} shows that these quantities can be computed using the matrix $M^c = (M^c_{\sigma\tilde\sigma})_{\sigma,\tilde\sigma\in\Sigma}$ given by
\[
M^c_{\sigma\tilde\sigma} =
\left\{
\begin{array}{cl}
\lambda_{\sigma\tilde\sigma} + c_\sigma 1_{\{\sigma = \tilde\sigma\}} & \text{ for } \sigma \in S, \\
1_{\{\sigma = \tilde\sigma\}} & \text{ for } \sigma \notin S,
\end{array}
\right.
\]
and the linear system
\begin{align}\label{eq:ec_linear_system}
\sum_{\tilde\sigma \in \Sigma} \lambda_{\sigma\tilde\sigma} x_{\tilde\sigma}+ c_\sigma x_\sigma = -1 & \text{ for } \sigma \in S, \\
\notag x_\sigma = 0 & \text{ for } \sigma \notin S.
\end{align}
If $M^c$ is invertible then~\eqref{eq:ec_linear_system} has a unique solution, which we denote $e^c = (e^c_\sigma)_{\sigma \in \Sigma}$. 

\begin{Lem}\label{lem:Ec_ec_general_properties} \mbox{}
	\begin{enumerate}[label=(\roman*)]
		\item If $S$ is irreducible then either $E^c_\sigma < +\infty$ for all $\sigma \in S$ or $E^c_\sigma = +\infty$ for all $\sigma \in S$.
		\item If $E^c_\sigma < +\infty$ for all $\sigma \in S$ then $E^c = (E^c_\sigma)_{\sigma \in \Sigma}$ is a solution of~\eqref{eq:ec_linear_system}. In particular, if $M^c$ is invertible then $E^c = e^c$.
		\item If $x = (x_\sigma)_{\sigma \in \Sigma}$ is a solution of~\eqref{eq:ec_linear_system} and $x_\sigma \ge 0$ for $\sigma \in S$ then $E^c_\sigma \le x_\sigma < +\infty$ for $\sigma \in S$.
	\end{enumerate}
\end{Lem}

This kind of result is classical, we include its proof for the convenience of the reader.

\begin{proof}
	(i) Let $0 = T_0 < T_1 < \cdots$ be the jump times of $\sigma_t$ and let $\sigma, \tilde \sigma \in S$ be arbitrary but fixed. Because $S$ is irreducible, there exist $\sigma = \varsigma_0, \varsigma_1, \ldots, \varsigma_N = \tilde\sigma \in S$ such that $\lambda_{\varsigma_{n}\varsigma_{n+1}} > 0$ for all $n < N$. Denoting $A = \cap_{n = 0}^N\{ \sigma_{T_n} = \varsigma_n \}$, the strong Markov property yields
	\begin{align*}
	\mathbb E_\sigma \left[ \int_0^{\tau} e^{\int_0^t c_{\sigma_s} ds} dt\right] &\ge \mathbb E_\sigma \left[ 1_A \int_0^{\tau} e^{\int_0^t c_{\sigma_s} ds} dt\right] \\
	&= \mathbb E_\sigma \left[ 1_A \int_0^{T_N} e^{\int_0^t c_{\sigma_s} ds} dt\right] + \mathbb E_\sigma \left[ 1_A e^{\int_0^{T_N} c_{\sigma_s} ds} \int_{T_N}^{\tau} e^{\int_{T_N}^{t} c_{\sigma_s} ds} dt\right] \\
	&= \mathbb E_\sigma \left[ 1_A \int_0^{T_N} e^{\int_0^t c_{\sigma_s} ds} dt\right] + \mathbb E_\sigma \left[ 1_A e^{\int_0^{T_N} c_{\sigma_s} ds} \right] \mathbb E_{\tilde\sigma} \left[ \int_0^{\tau} e^{\int_0^{t} c_{\sigma_s} ds} dt\right].
	\end{align*}
	
	Hence $E^c_{\tilde \sigma} = +\infty$ implies $E^c_\sigma = +\infty$. The assertion follows because $\sigma, \tilde\sigma\in S$ were arbitrary.
	
	(ii) Assume $E^c_\sigma < +\infty$ for all $\sigma \in S$. When $\sigma \notin S$ we have $E^c_\sigma = 0$. When $\sigma \in S$, conditioning on $T_1$ leads to
	\begin{align*}
	E^c_\sigma &= \mathbb E_\sigma \left[\int_0^{T_1} e^{\int_0^t c_{\sigma_s}ds} dt\right] + \mathbb E_\sigma \left[e^{\int_0^{T_1} c_{\sigma_s}ds} \right] \left(\sum_{\tilde\sigma\ne\sigma}  \frac{\lambda_{\sigma\tilde\sigma}}{\lambda_{\sigma}} \mathbb E_{\tilde\sigma} \left[\int_0^{\tau} e^{\int_0^t c_{\sigma_s}ds} dt\right] \right).
	\end{align*}
	The finiteness of $E^c_\sigma$ implies that $c_\sigma < \lambda_\sigma$ and
	\[
	\mathbb E_\sigma \left[\int_0^{T_1} e^{\int_0^t c_{\sigma_s}ds} dt\right] = \frac{1}{\lambda_\sigma - c_\sigma}
	\]
	hence
	\begin{align*}
	E^c_\sigma &= \frac{1}{\lambda_\sigma - c_\sigma} + \frac{\lambda_\sigma}{\lambda_\sigma - c_\sigma} \sum_{\tilde\sigma \ne \sigma} \frac{\lambda_{\sigma\tilde\sigma}}{\lambda_\sigma} E^c_{\tilde\sigma} \iff 	\sum_{\tilde\sigma \in \Sigma}\lambda_{\sigma\tilde\sigma} E^c_{\tilde\sigma} + c_\sigma E^c_\sigma = -1.
	\end{align*}
	
	(iii) Because $x$ solves \eqref{eq:ec_linear_system}, we have
	\begin{align*}
	x_\sigma &= \frac{1}{\lambda_\sigma - c_\sigma} + \frac{\lambda_\sigma}{\lambda_\sigma - c_\sigma} \sum_{\tilde\sigma \ne \sigma} \frac{\lambda_{\sigma\tilde\sigma}}{\lambda_\sigma} x_{\tilde\sigma} \\
	&= \frac{1}{\lambda_\sigma - c_\sigma} + \frac{\lambda_\sigma}{\lambda_\sigma - c_\sigma} \sum_{\tilde\sigma \ne \sigma} \frac{\lambda_{\sigma\tilde\sigma}}{\lambda_\sigma} \left( \frac{1}{\lambda_{\tilde\sigma} - c_{\tilde\sigma}} + \frac{\lambda_{\tilde\sigma}}{\lambda_{\tilde\sigma} - c_{\tilde\sigma}} \sum_{\hat\sigma \ne \tilde\sigma} \frac{\lambda_{\tilde\sigma\hat\sigma}}{\lambda_{\tilde\sigma}} x_{\hat\sigma} \right) \\
	&\ge \frac{1}{\lambda_\sigma - c_\sigma} + \frac{\lambda_\sigma}{\lambda_\sigma - c_\sigma} \sum_{\tilde\sigma \ne \sigma} \frac{\lambda_{\sigma\tilde\sigma}}{\lambda_\sigma} \frac{1}{\lambda_{\tilde\sigma} - c_{\tilde\sigma}} \\
	&= \mathbb E_\sigma \left[1_{\{\tau \le T_2\}}\int_0^{\tau} e^{\int_0^t c_{\sigma_s}ds} dt\right]
	\end{align*}
	using the non-negativity of the $x_\sigma$ for the inequality. Iterating the computation above yields
	\[
	x_\sigma \ge \mathbb E_\sigma \left[1_{\{\tau \le T_n\}}\int_0^{\tau} e^{\int_0^t c_{\sigma_s}ds} dt\right] \text{ for all } n \in \mathbb N.
	\]
	Hence it follows from the monotone convergence theorem that $E^c_\sigma \le x_\sigma < +\infty$.
\end{proof}

The following corollary is a consequence of the continuity of $c \mapsto \det M^c$ and $c \mapsto e^c_\sigma$. It is particularly useful when checking conditions~\eqref{eq:Ec_infinite} and~\eqref{eq:Ec_finite} of Theorem~\ref{thm:continuity_critical_point}.

\begin{Cor}\label{cor:finiteness_Ec_effective_condition} Assume that $M^a$ is invertible and that $S$ is irreducible.
	\begin{itemize}
		\item If $\min_{\sigma \in S} e^a_\sigma < 0$ then there exists $\gamma > 0$ such that
		\[
		\max_{\sigma \in S} |c_\sigma - a_\sigma| < \gamma \implies \min_{\sigma \in S} E^c_\sigma = +\infty.
		\]
		\item If $\min_{\sigma \in S} e^a_\sigma > 0$ then there exists $\gamma > 0$ such that
		\[
		\max_{\sigma \in S} |c_\sigma - a_\sigma| < \gamma \implies \max_{\sigma \in S} E^c_\sigma < +\infty.
		\]
	\end{itemize}
\end{Cor}

We illustrate this section's results by applying them to Counterexample~\ref{def:counterexample}.

\begin{Prop}\label{prop:counterexample_continuity_threshold}
	\begin{enumerate}[label=(\roman*)]
	\item In the case of Counterexample~\ref{def:counterexample}, taking $S = \{1, 2\}$, one has
	\[
	\eqref{eq:Ea_finite} \iff \omega > \frac{3+ \sqrt{5}}{2} \text{ for } \sigma = 1, 2.
	\]
	\item Furthermore
	\begin{align*}
	\eqref{eq:Ec_infinite} &\iff \omega < \frac{3+ \sqrt 5}{2}, & \eqref{eq:Ec_finite} &\iff \omega > \frac{3+ \sqrt 5}{2}.
	\end{align*}
		\end{enumerate}
\end{Prop}

\begin{proof}
	(i) Take $S = \{ 1, 2 \}$ and $a_1 = 1, a_2 = 2$. Using the notations $M^{c, \omega}, e^{c,\omega}_\sigma, E^{c, \omega}_\sigma$ instead of $M^c, e^c_\sigma, E^c_\sigma$ to keep track of the $\omega$ dependence, one has that
	\[
	M^{a,\omega} =
	\begin{pmatrix}
	-2\omega+ 1 & 2\omega & 0 \\
	\omega & -2\omega + 2 & \omega \\
	0 & 0 & 1
	\end{pmatrix}
	\]
	is invertible when $\omega \ne \frac{3 \pm \sqrt{5}}{2}$. In that case
	\begin{align*}
	e^{a,\omega}_1 &= \frac{2 {\omega} - 1}{{\omega}^{2} - 3 {\omega} + 1}, & e^{a,\omega}_2 &= \frac{3 {\omega} - 1}{2 {\left({\omega}^{2} - 3 {\omega} + 1\right)}}.
	\end{align*}
	It follows $E^{a, \omega}_\sigma = e^{a, \omega}_\sigma$ for $\omega> \omega^* := \frac{3 + \sqrt{5}}{2}$. Because $\omega \mapsto E^{a, \omega}_\sigma$ is non-increasing, as can be seen from stochastic domination and coupling, one has $E^{a, \omega}_\sigma = +\infty$ for $\omega \le \omega^*$.
	
	(ii) When $\omega > \omega^*$ it follows from Corollary~\ref{cor:finiteness_Ec_effective_condition} that there exists $\gamma(\omega) > 0$ such that
	\[
	\max_{\sigma \in S} |c_\sigma - a_\sigma| < \gamma(\omega) \implies \max_{\sigma \in S} E^{c,\omega}_\sigma < +\infty.
	\]
	Furthermore, there exists $\delta > 0$ such that if $\omega < \omega^*$ and $|\omega - \omega^*| < \delta$ then $M^{a,\omega}$ is invertible and $e^{a,\omega}_1, e^{a,\omega}_2 < 0$ so Corollary~\ref{cor:finiteness_Ec_effective_condition} implies the existence of $\gamma(\omega) > 0$ such that
	\[
	\max_{\sigma \in S} |c_\sigma - a_\sigma| < \gamma(\omega) \implies \min_{\sigma \in S} E^{c,\omega}_\sigma = +\infty.
	\]
	
	Finally, if $c_\sigma \ge 0$ for all $\sigma \in S$ then $\omega \mapsto E^{c,\omega}_\sigma$ is non-increasing for all $\sigma \in S$. Hence if $\omega \le \omega^* - \delta/2$ then
	\[
	\max_{\sigma \in S} |c_\sigma - a_\sigma| \le \min \left[\gamma({\omega^* - \delta/2}), \min_{\sigma \in S} \frac{a_\sigma}{2}\right] \implies \min_{\sigma \in S} E^{c,\omega}_\sigma \ge \min_{\sigma \in S} E^{c,\omega^*-\delta/2}_\sigma = +\infty.
	\]
\end{proof}

\section{Applications to the power-law and harmonic processes}\label{sec:applications}

This section is dedicated to the proof of Theorems~\ref{thm:instantaneous_powerlaw_process} and~\ref{thm:finite_harmonic_process}, which characterize the shape transition of the power-law process and the harmonic process respectively.

For the power-law process, we note that $v_0(0) = v_0'(0) = 0$. While it follows from~\cite[Th.~1]{balazs11} that $p_0$ is locally bounded at $x = 0$, continuity at that point cannot be studied using the results of~\cite{bakhtin15}, as they require $v_0'(0) \ne 0$. We use the following technical lemma to show continuity at $x = 0$ irrespective of model parameters through a direct computation. Its proof is postponed to the end of the section.

\begin{Lem}\label{lem:integral_eq}
	Let $\epsilon, \omega, a > 0$ and $p > 1$. If $R : \mathbb R \to \mathbb R_+$ is continuous at $0$ and $R(0) > 0$ then  
	\[
		\int_x^\epsilon e^{\frac{\omega}{a(p-1)}y^{-(p-1)}} R(y) dy \sim R(0) \frac{a}{\omega} e^{\frac{\omega}{a(p-1)}x^{-(p-1)}} x^p \text{ as } x \to {0+}.
	\]
\end{Lem}

\begin{proof}[Proof of Theorem~\ref{thm:instantaneous_powerlaw_process}] Using the terminology of~\cite[Sec.~4.2]{benaim23}, the point $x = 0$ is accessible and satisfies the weak bracket condition. It follows from~\cite[Th.~4.4]{benaim23} that the power-law process admits a unique invariant measure $\pi = \sum_{\sigma \in \Sigma} p_\sigma(x) dx \otimes \delta_\sigma$ and that this invariant measure admits a density. It follows from Theorem~\ref{thm:continuity_noncritical_interval} that
\begin{align*}
	p_2 &\in C^{0}(\mathbb R \setminus \{ x_+ \}), &
	p_0 &\in C^{0}(\mathbb R \setminus \{ 0 \}), &
	p_{-2} &\in C^{0}(\mathbb R \setminus \{ x_- \}).
\end{align*}
Using the terminology of~\cite[Sec.~6]{bakhtin15}, it follows from~\cite[Prop.~1]{bakhtin15} that $[x_{-}, x_{+}]$ is the only minimal invariant set of the process. The uniqueness of $\pi$ implies that it is ergodic. Hence it follows from~\cite[Prop.~7]{bakhtin15} that the support of the measure $p_\sigma(x) dx$ is $[x_{-}, x_{+}]$ for all $\sigma \in \Sigma$. It follows from~\cite[Th.~2]{bakhtin15} that
\[
p_2 \text{ is continuous at } x_+ \iff \lambda_2 > -v_2'(x_+) \iff \omega > a p x_+^{p-1}
\]
and similarly for the continuity of $p_{-2}$ at $x_-$.

It remains to show that $p_0$ is continuous at $x = 0$ irrespective of model parameters. It follows from Theorem~\ref{thm:Cr_regularity_noncritical_interval} that there exists $\epsilon > 0$ such that $p_\sigma \in C^\infty(0, \epsilon)$ for all $\sigma \in \Sigma$. Hence, setting $\varphi_0(x) = v_0(x) p_0(x)$, it follows from Lemma~\ref{lem:fokker_planck} that
\[
	\varphi_0'(x) = \frac{\omega}{a} x^{-p} \varphi_0(x) + \underbrace{2\omega p_2(x) + 2\omega p_{-2}(x)}_{:= R(x)}
\]
in the sense of classical ODEs. Explicitly solving this ODE yields that there exists $C \in \mathbb R$ such that
\begin{equation*}
\varphi_0(x) = C e^{-\frac{\omega}{a(p-1)}x^{-(p-1)}} - e^{-\frac{\omega}{a(p-1)}x^{-(p-1)}} \int_x^\epsilon e^{\frac{\omega}{a(p-1)}y^{-(p-1)}} R(y) dy \text{ for } x \in (0, \epsilon).
\end{equation*}
Hence
\[
p_0(x) = \underbrace{\frac{C e^{-\frac{\omega}{a(p-1)}x^{-(p-1)}}}{-2ax^p}}_{\overset{x \to 0+}{\longrightarrow} 0}
 + \frac{ e^{-\frac{\omega}{a(p-1)}x^{-(p-1)}}}{2a x^p} \int_x^\epsilon e^{\frac{\omega}{a(p-1)}y^{-(p-1)}} R(y) dy.
\]

Because $p_{-2}, p_{2}$ are continuous at $0$, so is $R$. One has $p_2(0), p_{-2}(0) > 0$ by~\cite[Th.~4.4]{benaim23} so $R(0) > 0$. Hence it follows from Lemma~\ref{lem:integral_eq} that, in the $x \to {0+}$ limit,
\[
	\frac{ e^{-\frac{\omega}{a(p-1)}x^{-(p-1)}}}{2a x^p} \int_x^\epsilon e^{\frac{\omega}{a(p-1)}y^{-(p-1)}} R(y) dy \sim \frac{ e^{-\frac{\omega}{a(p-1)}x^{-(p-1)}}}{2a x^p} \frac{a}{\omega} e^{\frac{\omega}{a(p-1)}x^{-(p-1)}} x^p R(0) \sim \frac{R(0)}{2\omega}.
\]
Hence $\lim_{x \to 0+} p_0(x) = \frac{R(0)}{2\omega}$ and the same argument shows $\lim_{x \to 0-} p_0(x) = \frac{R(0)}{2\omega}$.
\end{proof}

The main difficulty in understanding the shape transition of the harmonic process is the joint vanishing of $v_{0_\pm}$ and $v_{0_0}$ at $x = 0$. This is addressed using Theorem~\ref{thm:continuity_critical_point}. 

\begin{proof}[Proof of Theorem~\ref{thm:finite_harmonic_process}] As in the proof of Theorem~\ref{thm:instantaneous_powerlaw_process}, by~\cite[Sec.~4.2]{benaim23}, the fact that the point $x = 0$ is accessible and satisfies the weak bracket condition implies that the harmonic process admits a unique invariant measure $\pi = \sum_{\sigma \in \Sigma} p_\sigma(x) dx \otimes \delta_\sigma$ and that this invariant measure possesses a density. It follows from Theorem~\ref{thm:continuity_noncritical_interval} that
\begin{align*}
	p_{\pm k} &\in C^0 \left( \mathbb R \setminus \{ x_{\pm k} \} \right), & p_{0_\pm}, p_{0_0} &\in C^0(\mathbb R \setminus \{ 0 \}).
\end{align*}
As in the proof of Theorem~\ref{thm:instantaneous_powerlaw_process}, $\pi$ is ergodic and $[x_{-2}, x_{+2}]$ is the only minimal invariant set of the process. Hence it follows from~\cite[Prop.~1]{bakhtin15} and~\cite[Prop.~7]{bakhtin15} that the support of the measure $p_\sigma(x)dx$ is $[x_{-2}, x_{+2}]$ for all $\sigma \in \Sigma$. The continuity or lack thereof of $p_{\pm k}$ at $x = x_{\pm k}$ follows from~\cite[Th.~2]{bakhtin15}. It remains to discuss the continuity of $p_{0_\pm}, p_{0_0}$ at $x = 0$ using Theorem~\ref{thm:continuity_critical_point}.

Take $S = \{ 0_\pm \}$. It is immediate that Assumptions~\ref{ass:D1_compact_set}, \ref{ass:D2_right_diff}, \ref{ass:D3_S_irreducible} and~\ref{ass:E1_analytic} are satisfied. Assumption~\ref{ass:D4_TV_convergence} can be checked using~\cite[Cor.~2.7]{benaim18}. Moreover, because the support of $p_{0_\pm}(x) dx$ is $[x_{-2}, x_{+2}]$, Assumption~\ref{ass:D5_inside_support} is also satisfied. One has that
\[
B_0 = 
\bordermatrix{ & \tilde \sigma = 2 & \tilde \sigma = 1 & \tilde \sigma = 0_\pm & \tilde \sigma = 0_0 & \tilde \sigma = -1 & \tilde \sigma = -2 \cr
	\sigma = 2     & 0 & 0 & 0                  & 0 & 0 & 0 \cr
	\sigma = 1     & 0 & 0 & -\frac{\alpha}{2a} & -\frac{\beta}{2a} & 0 & 0 \cr
	\sigma = 0_\pm & 0 & 0 & \frac{\alpha}{a}   & 0 & 0 & 0 \cr
	\sigma = 0_0   & 0 & 0 & 0                  & \frac{\beta}{a} & 0 & 0 \cr
	\sigma = -1    & 0 & 0 & -\frac{\alpha}{2a} & -\frac{\beta}{2a} & 0 & 0 \cr
	\sigma = -2    & 0 & 0 & 0                  & 0 & 0 & 0 
},
\]
is diagonalizable and its eigenvalues $\alpha/a, \beta/a, 0$ are all real. Thus Assumption~\ref{ass:E2_spectral} holds. Finally,
\[
A = 
\bordermatrix{  & \tilde \sigma = 1 & \tilde \sigma = 0_\pm  & \tilde \sigma = -1  \cr
	\sigma = 1 & 1 & 0 & 0  \cr
	\sigma = 0_\pm& \beta/2 & -2\alpha + 2a & \beta/2  \cr
	\sigma = -1 & 0 & 0 & 1  \cr
},
\]
is invertible, i.e.~Assumption~\ref{ass:E3_invertibility} is satisfied, when $\alpha \ne a$. Because $S = \{0_\pm\}$ is a singleton, we have
\begin{align*}
	\eqref{eq:Ec_infinite} &\iff \lambda_{0_\pm} < -v_{0_\pm}'(0) \iff \alpha < a, \\
	\eqref{eq:Ec_finite} &\iff \lambda_{0_\pm} > -v_{0_\pm}'(0) \iff \alpha  > a.
\end{align*}
Theorem~\ref{thm:continuity_critical_point} yields that $p_{0_\pm}$ is continuous (resp.~diverges) at $x = 0$ if $\alpha> a$ (resp.~$\alpha < a$). It follows from the same argument that $p_{0_0}$ is continuous (resp.~diverges) at $x = 0$ if $\beta > a$ (resp.~$\beta< a$).
\end{proof}

We end with the postponed proof of the technical lemma.

\begin{proof}[Proof of Lemma~\ref{lem:integral_eq}]
	Set $q = p - 1$ and split the integral as follows
	\[
	\int_{x}^\epsilon e^{\frac{\omega}{aq}y^{-q}} R(y) dy = \int_x^{2x} e^{\frac{\omega}{aq}y^{-q}} R(y) dy + \int_{2x}^\epsilon e^{\frac{\omega}{aq}y^{-q}} R(y) dy.
	\]
	One has
	\[
	\int_x^{2x} e^{\frac{\omega}{aq}y^{-q}} R(y) dy \ge \int_x^{\frac{3}{2}x} e^{\frac{\omega}{aq}y^{-q}} R(y) dy \ge \frac{1}{2} x e^{\frac{\omega}{aq}\left(\frac{3}{2}x\right)^{-q}} \inf_{y \in [x, 3x/2]} R(y)
	\]
	and
	\[
	\int_{2x}^\epsilon e^{\frac{\omega}{aq}y^{-q}} R(y) dy \le (\epsilon - 2x) e^{\frac{\omega}{aq}\left( 2x \right)^{-q}} \sup_{y \in [2x, \epsilon]} R(y).
	\]
	It follows
	\[
	\int_{x}^\epsilon e^{\frac{\omega}{aq}y^{-q}} R(y) dy \sim \int_{x}^{2x} e^{\frac{\omega}{aq}y^{-q}} R(y) dy \sim R(0) \int_{x}^{2x} e^{\frac{\omega}{aq}y^{-q}} dy \text{ as } x \to {0+}.
	\]
	Integrating by parts yields
	\[
	\int_{x}^{2x} e^{\frac{\omega}{aq}y^{-q}} dy = \int_{x}^{2x} \frac{e^{\frac{\omega}{aq}y^{-q}}}{y^p} y^pdy = \left[-\frac{a}{\omega} e^{\frac{\omega}{aq}y^{-q}} y^p\right]_{x}^{2x} + \int_x^{2x} \frac{a}{\omega} e^{\frac{\omega}{aq}y^{-q}} p y^q dy.
	\]
	One has
	\[
	\frac{\int_x^{2x} \frac{a}{\omega} e^{\frac{\omega}{aq}y^{-q}} p y^q dy}{\int_x^{2x} \frac{a}{\omega} e^{\frac{\omega}{aq}y^{-q}} dy} \to 0 \text{ as } x \to {0+}
	\]
	so that
	\[
	\int_{x}^{2x} e^{\frac{\omega}{aq}y^{-q}} dy \sim \left[-\frac{a}{\omega} e^{\frac{\omega}{aq}y^{-q}} y^p\right]_{x}^{2x} \sim \frac{a}{\omega} e^{\frac{\omega}{aq}x^{-q}} x^p.
	\]
\end{proof}

\section*{Funding}
The author is supported by the grant 200029\_21991311 of the Swiss National Science Foundation.

\section*{Use of AI tools}

Claude was used for proofreading, LaTeX formatting and style improvement.

\section*{Acknowledgments}
The author would like to sincerely thank Michel Bena\"im and Tobias Hurth for many helpful discussions.

\bibliographystyle{alpha}
\bibliography{biblio}

\end{document}